\pdfoutput=1
\documentclass[ejs,twoside,preprint]{imsart}

\usepackage{amsbsy,amsmath,amsthm,amssymb,graphicx,verbatim}
\usepackage{subfigure}
\usepackage[longnamesfirst,sort]{natbib}
\newcommand{\pcite}[1]{\citeauthor{#1}'s \citeyearpar{#1}}
\usepackage{multirow}

\RequirePackage[colorlinks,citecolor=blue,urlcolor=blue]{hyperref}
\usepackage{breakurl}

\doi{10.1214/11-EJS663}
\pubyear{2012}
\volume{6}
\issue{0}
\firstpage{10}
\lastpage{37}

\newtheorem{theorem}{Theorem}

\theoremstyle{definition}
\newtheorem{remark}{Remark}

\newtheorem{algorithm}{Algorithm}

\newcommand{\sX}{\mathsf{X}}
\newfont{\msbm}{msbm10 at 11pt}
\newcommand {\R} {\mathbb{R}}

\newcommand {\N} {\mbox{\msbm N}}

\newcommand{\Var}[0]{\text{Var}}

\newcommand{\beq}{\begin{eqnarray*}}
\newcommand{\eeq}{\end{eqnarray*}}

\newcommand{\beqn}{\begin{eqnarray}}
\newcommand{\eeqn}{\end{eqnarray}}

\begin{document}

\begin{frontmatter}

\title{Exact sampling for intractable probability distributions via a Bernoulli factory}
\runtitle{Exact sampling via a Bernoulli factory}

\begin{aug}
\author{\fnms{James M.} \snm{Flegal}\corref{}\ead[label=e1]{jflegal@ucr.edu}}

\address{Department of Statistics \\ University of California, Riverside\\
           \printead{e1}}
\end{aug}
%\medskip
\vspace*{6pt}
\textbf{\and}
\vspace*{-6pt}
\begin{aug}
\author{\fnms{Radu} \snm{Herbei}\ead[label=e3]{herbei@stat.osu.edu}}

\address{Department of Statistics \\ The Ohio State University\\
          \printead{e3}}

  \runauthor{J.M. Flegal and R. Herbei}

\end{aug}

\begin{abstract}
Many applications in the field of statistics require Markov chain Monte
Carlo methods. Determining appropriate starting values and run lengths
can be both analytically and empirically challenging. A desire to
overcome these problems has led to the development of exact, or
perfect, sampling algorithms which convert a Markov chain into an
algorithm that produces i.i.d.\ samples from the stationary
distribution. Unfortunately, very few of these algorithms have been
developed for the distributions that arise in statistical applications,
which typically have uncountable support.  Here we study an exact
sampling algorithm using a geometrically ergodic Markov chain on a
general state space.  Our work provides a significant reduction to the
number of input draws necessary for the Bernoulli factory, which
enables exact sampling via a rejection sampling approach.  We
illustrate the algorithm on a univariate Metropolis-Hastings sampler
and a bivariate Gibbs sampler, which provide a proof of concept and
insight into hyper-parameter selection.  Finally, we illustrate the
algorithm on a Bayesian version of the one-way random effects model
with data from a styrene exposure study.
\end{abstract}



% \tableofcontents
\received{\smonth{12} \syear{2010}}

\end{frontmatter}

%s1 ###
\section{Introduction} \label{sec:intro}

Suppose we want to explore a probability distribution $\pi$ defined on
$\sX$.  Further suppose $\pi$ is intractable in the sense that direct
(i.i.d.) sampling is unavailable.  In this setting, the Markov chain
Monte Carlo (MCMC) method can be a useful tool since it is often
straightforward to construct and simulate an ergodic Markov chain that
has $\pi$ as its stationary distribution
\citep{chen:shao:ibra:2000,robe:case:1999,liu:2001}.  The two main
drawbacks of MCMC relative to direct sampling from $\pi$ are (i)
the\vadjust{\eject}
difficulty in ascertaining how long the Markov chain needs to be run
before it gets ``close'' to $\pi$ \citep[see, e.g.,][]{jone:hobe:2001},
and (ii) the difficulty in deriving and calculating asymptotically
valid standard errors for the ergodic averages that are used to
approximate intractable expectations under $\pi$ \citep[see,
e.g.,][]{fleg:hara:jone:2008}.

% The Markov chain Monte Carlo (MCMC) method is a useful tool for solving problems in the field of statistics because of the following fact.  It is often straightforward to construct and simulate a Markov chain that converges to a given intractable probability distribution, $\pi$, even in cases where direct (i.i.d.) sampling from $\pi$ is impossible \citep{chen:shao:ibra:2000,robe:case:1999,liu:2001}.  The two main drawbacks of MCMC relative to direct sampling from $\pi$ are (i) the difficulty in ascertaining how long the Markov chain needs to be run before it gets ``close'' to $\pi$ \citep[see, e.g.,][]{jone:hobe:2001}, and (ii) the difficulty in deriving and calculating asymptotically valid standard errors for the ergodic averages that are used to approximate intractable expectations under $\pi$ \citep[see, e.g.,][]{fleg:hara:jone:2008}.

A desire to overcome these problems has led to the development of
clever techniques using a Markov chain to create an algorithm that
produces i.i.d.\ draws from the stationary distribution
(e.g., \citealp{crai:meng:2011}; Green and Murdoch, \citeyear{gree:murd:1999}; \citealp{hube:2004, prop:wils:1996, wils:2000}).
Unfortunately, very few of these so-called perfect
sampling algorithms have been developed for the distributions that
arise in realistic statistical applications, which typically have
uncountable support.  \cite{asmu:glyn:thor:1992} and
Blanchet and Meng (\citeyear{blan:meng:2005}) provide one such algorithm applicable to Markov
chains on general state spaces.  The main assumption necessary is that
the chain satisfies a one-step minorization condition.  As we describe
later, under this condition the stationary distribution admits a {\em
mixture representation}, suggesting the following two-step sampling
approach: sample the discrete distribution corresponding to the mixture
weights, then sample the selected mixture component.

This approach has never been successfully implemented, however, it has been used to obtain approximate draws from $\pi$, see for example \citet{blan:thom:2007}, \citet{hobe:jone:robe:2006} and \citet{hobe:robe:2004}.  The difficult part is drawing from the discrete distribution corresponding to the mixture weights, which is done via a rejection sampling approach.  In this paper, we provide solutions to a number of practical problems and illustrate the algorithm on three examples.  This requires overcoming two challenges: (i) obtaining a dominating proposal distribution and (ii) generating a ${\rm Bernoulli}$ variate to decide whether a proposed draw is accepted or not.  While (ii) might seem trivial, the challenge is that we are unable to (exactly) compute the success probability for this ${\rm Bernoulli}$
variate.\looseness=-1

A solution to (i) requires identification of a bounding (proposal) probability mass function for the target mass function.  Previously, \citet{blan:meng:2005} proposed an upper bound on moments associated with the target mass function.  \citet{blan:thom:2007} used output from a preliminary run of the Markov chain to construct an approximate upper bound.  We provide an explicit bound for Markov chains satisfying a geometric drift condition using results from \citet{robe:twee:1999}.

A solution to (ii) will determine whether to accept a proposed draw.  The decision is made by generating a ${\rm Bernoulli}$ random variable with success probability that involves the ratio between the target and proposal mass functions.  In our case, the target mass function is unknown, apparently making this step impossible.  However, one can still generate such a ${\rm Bernoulli}$ variate, using a so-called {\em Bernoulli factory} \citep{kean:obri:1994}.  Briefly, a Bernoulli factory is an algorithm that outputs a ${\rm Bernoulli}$ variate with success probability $f(p)$, from  i.i.d.\ ${\rm Bernoulli}(p)$ variates, when $f$ is known but $p$ is unknown.
%\vfill\eject

The rejection sampling approach requires a Bernoulli factory algorithm for $f(p)=ap$ and $a \in (1 , \infty)$.  \cite{nacu:pere:2005} and \cite{latu:kosm:papa:robe:2011} provide an algorithm when $a = 2$, but their algorithms are computationally demanding and scale poorly for $a \in (1 , \infty)$.  For example when $p \in (0 , .4)$, one requires at least $65,536$ ${\rm Bernoulli}(p)$ random variables to generate one ${\rm Bernoulli}(2p)$ variate.  In this paper, we provide an algorithm for any $a \in (1 , \infty)$ that reduces the computational time substantially.  For example when $p \in (0, .4)$, we can obtain a ${\rm Bernoulli}(2p)$ variate with only $256$ ${\rm Bernoulli}(p)$ random variables.  This is an important reduction because the Bernoulli factory accounts for much of the computational time in the exact sampling algorithm.  Section~\ref{sec:Bernoulli} contains a full description the Bernoulli factory and our modification.

Our solutions to (i) and (ii) yield an exact sampling algorithm for $\pi$.  The algorithm is suitable even for intractable distributions on general state spaces: that is, for distributions that typically arise in statistical applications. This is an important extension, since very few existing algorithms apply to general state spaces, but it is limited in the sense that one must be able to establish a drift and associated minorization condition for the underlying Markov chain.  The current algorithm can be computationally demanding, however we have successfully implemented it in three examples.

Our first example considers a univariate Metropolis-Hastings sampler for which we obtain 1000 i.i.d.\ draws.  The second example considers a slightly more complicated bivariate Gibbs sampler where we again obtain 1000 i.i.d.\ draws.  These two examples could be considered toy examples in the sense that i.i.d.\ observations are available for each.  However, they provide insights into the performance and hyper-parameter selection of the algorithm.

Our final example considers a Bayesian version of the classical one-way random effects model that is widely used to analyze data.  We illustrate the exact sampling algorithm, using data from a styrene exposure study, to obtain 20 i.i.d.\ draws.  This is the first successful implementation of an exact sampling algorithm for a model of this type.  Our analysis considers a balanced design and requires development of a suitable drift condition, which improves upon the existing drift constants of \cite{tan:hobe:2009}.

These examples give hope for exact sampling algorithms for general state space Markov chains in more complicated settings.  Even if we are unable to obtain multiple draws in these settings, a single exact draw will alleviate the need for burn-in entirely.

The rest of this paper is organized as follows.  Section~\ref{sec:rejection} provides the mixture representation of $\pi$, details the rejection sampling approach and bounds the tail probabilities of the proposal distribution.  Section~\ref{sec:Bernoulli} introduces the Bernoulli factory and proposes a new target function that speeds up the algorithm significantly.  Section~\ref{sec:exact} gives the full exact sampling algorithm.  Sections~\ref{sec:toy} and~\ref{sec:real} implement the algorithm for two toy examples and a Bayesian version of a one-way random effects model, respectively.  Finally, Section~\ref{sec:dis} discusses our implementation and provides some general recommendations to practitioners.
%\vfill

\vspace*{-3pt}
%s2 ###
\section{Exact sampling via a mixture distribution} \label{sec:rejection}
\vspace*{-3pt}

Suppose we want to explore the intractable probability measure $\pi(dx)$ defined on the measurable space $\left( \sX,{\cal B}(\sX) \right)$.  Let $P: \sX \times {\cal B}(\sX) \rightarrow [0,1]$ be a Markov transition function and let $X = \{X_n\}_{n=0}^\infty$ denote the corresponding Markov chain.  Then for $x \in \sX$ and a measurable set $A$,
\[
P(x,A) = \Pr \left( X_{n+1} \in A|X_n = x \right) \;.
\]
Assume that $\pi$ is an invariant measure for the chain; i.e., $\pi(A) = \int_\sX P(x,A) \, \pi(dx)$ for all measurable $A$.  Assume further that $X$ satisfies the usual regularity conditions, which are irreducibility, aperiodicity and positive Harris recurrence.  For definitions, see \citet{meyn:twee:1993} and \citet{robe:rose:2004}.  Finally, assume we are able to simulate the chain; that is, given $X_n=x$, we have the ability to draw from $P(x,\cdot)$.

The exact sampling algorithm considered here utilizes a mixture representation for $\pi$ \citep{hobe:jone:robe:2006, asmu:glyn:thor:1992, hobe:robe:2004}, however, we must first develop the split chain.  The main assumption necessary is that $X$ satisfies a one-step minorization condition, i.e.\ there exists a function $s : \sX \rightarrow [0,1]$ satisfying $\int_{\sX} s(x) \, \pi(dx)>0$ and some probability measure $Q(dy)$ on $(\sX, {\cal B}(\sX))$ such that,
\begin{equation} \label{eq:mc}
 P(x,A) \ge s(x) \, Q(A) \;\; \mbox{for all} \; x \in \sX \; \mbox{and} \; A \in {\cal B}(\sX) \;.
\end{equation}

Given $X$ satisfies the one-step minorization condition at \eqref{eq:mc}, then $P$ can be decomposed as
\begin{equation}
 \label{eq:split}
 P(x,dy) = s(x) \, Q(dy) + \left( 1-s(x) \right) \, R(x,dy) \;,
\end{equation}
where
\[
R(x,dy) = \frac{P(x,dy) - s(x) \, Q(dy)}{1-s(x)} \; ,
\]
and define $R(x, dy) = 0$ if $s(x) = 1$.  It is helpful to think of \eqref{eq:split} as a mixture of two Markov transition functions with probabilities $s(x)$ and $1-s(x)$.  Equation \eqref{eq:split} shows that it is possible to simulate $X_{n+1}$ given $X_n=x$ as follows: Flip a coin (independently) that comes up heads with probability $s(x)$.  If the coin is a head, take $X_{n+1} \sim Q(\cdot)$; if it's a tail, take $X_{n+1} \sim R(x,\cdot)$.

This decomposition has several important applications in MCMC.  Indeed,
it can be used to perform regenerative simulation
(\citealp{hobe:jone:pres:rose:2002}; Mykland et~al., \citeyear{mykl:tier:yu:1995}) and to derive
computable bounds on the convergence rate of $X$
\citep{rose:2002,rose:1995a,lund:twee:1996,robe:twee:1999}.

Now consider a new Markov chain that actually includes the coin flips mentioned above.  Let $X' = \{(X_n,\delta_n)\}_{n=0}^\infty$ be a Markov chain with state space $\sX \times \{0,1\}$.  If the current state is $(X_n,\delta_n) = (x,\delta)$, then the next state, $(X_{n+1},\delta_{n+1})$, is drawn as follows.  If $\delta=1$, then $X_{n+1} \sim Q(\cdot)$; while if $\delta=0$, $X_{n+1} \sim R(x,\cdot)$.  Then, conditional on $X_{n+1}=x'$, $\delta_{n+1} \sim \mbox{Bernoulli}(s(x'))$.  This chain is called the split chain.  Equation \eqref{eq:split} implies that, marginally, the sequence of $X_n$ values in the split chain has the same overall probability law as the original Markov chain $X$ \citep[Chapter 4]{numm:1984}.  Note that, if $\delta_n=1$, then the distribution of $(X_{n+1},\delta_{n+1})$ does not depend on $x$.  %We will always take the starting value $(X_0,\delta_0)$ such that $\delta_0=1$, and consequently, which means that $X_1 \sim Q(\cdot)$ and we don't even have to specify the distribution of $X_0$.

\begin{remark}
We can avoid drawing from $Q(\cdot)$ entirely by changing the order slightly.  Given $x$ is the current state, we can simply generate $X_{i+1} \sim P(x, \cdot)$ in the usual manner and then generate $\delta_i | X_{i}, X_{i+1}$ with
\begin{equation} \label{eq:delta equals 1}
Pr \left( \delta_i = 1 | X_{i}, X_{i+1} \right) = \frac{s(X_{i}) q(X_{i+1})}{k(X_{i+1} | X_{i})} \; ,
\end{equation}
where $q(\cdot)$ and $k(\cdot | x)$ are the densities corresponding to $Q(\cdot)$ and $P$ \citep[][p.\ 62]{numm:1984}.  %In many situations, $s(x) > 0$ only on some small set $C$, and hence the calculation at \eqref{eq:delta equals 1} is only necessary if $x \in C$.
\end{remark}

%s2.1 ###
\subsection{A mixture representation of $\pi$} \label{sec:mixture}

The reason for introducing the split chain is that it possesses an accessible atom, $\sX \times \{1\}$.  Indeed, each time the set $\sX \times \{1\}$ is entered, the split chain stochastically restarts itself (because the next $X_n$ has distribution $Q$).  Let $\mathbb{N} = \{1,2,3,\dots\}$ and $ (X_0,\delta_0) \in \sX \times \{1\}$, then define the first return time to the atom as
\[
\tau = \min \big \{ n \in \mathbb{N} : (X_n,\delta_n) \in \sX \times \{1\} \big \} \; .
\]
Our assumptions about $P$ imply that $\mbox{E}(\tau) < \infty$ and hence the sequence $\{p_n\}_{n=1}^\infty$ defined by
\begin{equation*}
p_n = \frac{\Pr(\tau \ge n)}{\mbox{E}(\tau)}
\end{equation*}
is nonnegative, nonincreasing, and sums to one.  Let $T$ denote a discrete random variable on $\mathbb{N}$ with $\Pr(T=n) = p_n$; also let $Q_n$ be the conditional distribution of $X_n$ given that the split chain does not return to the atom before time $n$.  Thus for any $n \in \mathbb{N}$ and any measurable $A$, $Q_n(A) = \mbox{Pr}(X_n \in A \mid \tau \ge n )$.  (Note that $Q_1 \equiv Q$.)  Then $\pi$ can be written as the following mixture of the $Q_n$ values:
\begin{equation} \label{eq:mix_rep}
\pi(dx) = \sum_{n=1}^\infty p_n \, Q_n(dx) \;.
\end{equation}

\begin{remark}
The representation at \eqref{eq:mix_rep} can be obtained from results
in Asmussen et~al. (\citeyear{asmu:glyn:thor:1992}) by applying their methods to the split
chain.  Alternatively, \citet{hobe:robe:2004} obtain the representation
when $s(x)$ has the specific form $\varepsilon I_{C} (x)$ and
\citet{hobe:jone:robe:2006} obtain the representation with the more
general minorization shown here.
\end{remark}

The representation at \eqref{eq:mix_rep} offers an alternative sampling
scheme for $\pi$.  First, make a random draw from the set $\left\{ Q_1,
Q_2, Q_3, \dots \right\}$ according to the probabilities $p_1, p_2,
p_3, \dots $ and then make an independent random draw from the
chosen~$Q_n$. %\vfill\eject

%\vspace*{-18pt}
\noindent\hrulefill

\noindent {\rm Sampling algorithm for $\pi$:}
\vspace*{-3pt}
\begin{enumerate}
\item Draw $T$ such that $Pr(T = n) = p_n$ for $n = 1, 2, 3, \dots$, call the result $t$.
\item Make a draw from $Q_t(\cdot)$.
\end{enumerate}

\vspace*{-4mm}
\noindent\hrulefill
\eject

Drawing from $Q_n$ is simple even when $n \ge 2$.  Indeed, just repeatedly simulate $n$ iterations of the split chain, and accept $X_n$ the first time that $\delta_1=\cdots=\delta_{n-1}=0$.  The challenging part of this recipe is drawing from the set $\left\{ Q_1, Q_2, Q_3, \dots \right\}$, i.e.\ simulating a random variable $T$, since the $p_n$ values are not computable.  \citet{hobe:jone:robe:2006} and \citet{hobe:robe:2004} approximate the $p_n$ values, which yields approximate draws from $T$ and thus $\pi$.  In this paper, we obtain exact draws from $T$ that result in exact draws from $\pi$.

\begin{remark}
Let $K | T=n$ be the number of simulations of the split chain before we get a draw from $Q_n$, conditional on $T=n$ from Step 1 above.  Then $K | T=n$ is geometric with mean $\mbox{E}( K | T=n ) = 1 / P(\tau \ge n ) $.  Unfortunately, \citet{blan:meng:2005} show $\mbox{E}( K) = \infty$, and justifiably argue \eqref{eq:mix_rep} should not be used for multiple replications.  This presents a major challenge in the applicability of our algorithm and others that can be similarly expressed, some of which are discussed in the next section.
\vspace*{-3pt}
\end{remark}

%s2.2 ###
\subsection{Rejection sampler for $T$}
\vspace*{-3pt}

There is one case where simulating $T$ is simple
(\citealp{hobe:jone:robe:2006}; Hobert and Robert, \citeyear{hobe:robe:2004}).  Suppose that in the minorization condition
\eqref{eq:mc}, $s(x) \equiv \varepsilon > 0$ for all $x \in \sX$
(implying the Markov chain is uniformly ergodic) and consider the
procedure for simulating the split chain with this constant $s$.  In
particular, note that the coin flip determining whether $\delta_n$ is 0
or 1 does not depend on $x$, and it follows that the number of steps
until the first return to the accessible atom has a geometric
distribution.  Indeed, $\Pr(\tau \ge n) = (1 - \varepsilon)^{n-1}$.
Hence, $\mbox{E}(\tau) = 1/\varepsilon$ and
\[
p_n = \Pr(T = n) = \frac{\Pr(\tau \ge n)}{\mbox{E}(\tau)} = \varepsilon (1 - \varepsilon)^{n-1} \;,
\]
so $T$ also has a geometric distribution.  Therefore it is easy to make exact draws from $\pi$.

\citet{hobe:robe:2004} show this exact sampling algorithm is equivalent to \pcite{murd:gree:1998} Multigamma Coupler and to \pcite{wils:2000} Read-Once algorithm.  It is interesting that \eqref{eq:mix_rep} can be used to reconstruct perfect sampling algorithms based on coupling from the past despite the fact that its derivation involves no backward simulation arguments.  Of course, this exact sampling algorithm will be useless from a practical standpoint if $\varepsilon$ is too small.

Unfortunately, in statistical inference problems, the MCMC algorithms are usually driven by Markov chains that are not uniformly ergodic and, hence, cannot satisfy \eqref{eq:mc} with a constant $s$.  Moreover, there is no efficient method to simulate $T$ where $s$ is non-constant.  (When $s$ is non-constant, the distribution of $\tau$ is complex and its mass function is not available in closed form. Hence, the mass function of $T$ is also unknown, which precludes direct simulation of $T$.)  Therefore we must resort to indirect methods of simulating $T$.

Fortunately, simulating $\tau$ is trivial---indeed, one can simply run the split chain and count the number of steps until it returns to $\sX \times \{1\}$.  Because this provides an unlimited supply of i.i.d.\ copies of $\tau$, we can use a rejection sampling approach \citep{asmu:glyn:thor:1992, blan:meng:2005} to simulate $T$ from the i.i.d.\ sequence $\tau_1,\tau_2,\dots$ (where $\tau_1 \stackrel{d}{=} \tau$).

Suppose there exists a function $d: \mathbb{N} \rightarrow [0,1]$ such that $\sum_{n=1}^\infty d(n) = D < \infty$ and $P(\tau \ge n) \le M d(n)$ where $M$ is a finite, positive constant. Consider a rejection sampler with candidate mass function $d(\cdot)/D$.  Thus
\[
\frac{\Pr(T = n)}{d(n)/D} = \frac{\Pr(\tau \ge n) / \mbox{E}(\tau)}{d(n)/D} = \frac{D}{\mbox{E}(\tau)} \frac{\Pr(\tau \ge n)}{d(n)} \le \frac{D}{\mbox{E}(\tau)} M \; ,
\]
which justifies the following rejection sampler.

\noindent\hrulefill

\noindent {\rm Rejection sampler for simulating $T$:}
\begin{enumerate}
\item Draw $T \sim \frac{d(\cdot)}{D}$.  Call the result $n$ and let $a = 1 / \left[ M d(n) \right]$.
\item Draw an independent Bernoulli random variable, $B$, with success probability $ a \Pr(\tau \ge n)$.  If $B=1$, accept $T=n$; if $B=0$, return to Step~1.
\end{enumerate}

\vspace*{-4mm}
\noindent\hrulefill

Unfortunately, the standard method of simulating $B$ (by computing $ a \Pr(\tau \ge n)$ and comparing it to an independent $\mbox{Uniform}(0,1)$ random variable) is not available to us because the mass function of $\tau$ is unavailable in closed form.  However, we may draw $B$ without knowing the value of $\Pr(\tau \ge n)$ using a supply of i.i.d.\ copies of $\tau$. This is the basis of our exact sampling approach.

Suppose $a \in (0 , 1 ]$ and let $p = \Pr(\tau \ge n)$, then there exists a simple solution to generate $B\sim {\rm Bernoulli}(a p)$, which \citet{fill:1998} calls ``engineering a coin flip''.  Indeed simulate a single $\tau$ and define
\begin{equation*}
W = \begin{cases} 1 \quad & \text{if } \tau \ge n  \\
        0 \quad & \text{if } \tau < n \end{cases}  \; ,
\end{equation*}
hence $W \sim \mbox{Bernoulli}(p)$.  If we independently simulate $V \sim \mbox{Bernoulli}(a)$ as usual and set $B =V W$, then
\[
\Pr \left( B = 1 \right) = \Pr \left( [V = 1] \cap [W = 1] \right) = \Pr \left( V = 1 \right) \Pr \left( W = 1 \right) = a p\; .
\]
That is, $B\sim {\rm Bernoulli}(a p)$, obtained by simulating a single $\tau$ and a single Bernoulli $V$.

When $a \in (1 , \infty)$, we will obtain $B\sim {\rm Bernoulli}(a p)$ via the Bernoulli factory described in Section~\ref{sec:Bernoulli}.  For now assume such a simulation is possible, then what remains to establish is a computable tail probability bound, i.e. the sequence $d(n)$ and the constant $M$.

%s2.3 ###
\subsection{Tail probability bound} \label{sec:bound and algorithm}
\cite{blan:meng:2005} bound the moments of $\tau$, however they do not explicitly determine computable values $M$ and $d(n)$.  Fortunately, $\Pr(\tau \ge n)$ can be bounded above by a known constant times a known geometric mass function if $X$ satisfies a geometric drift and associated one-step minorization conditions.  We will say a drift condition holds if there exists some function $V: \sX \mapsto [1, \infty)$, some $0 < \lambda < 1$ and some $b < \infty$, such that
\begin{equation} \label{eq:drift}
E\left[  V(X_{i+1}) | X_i = x \right]  \leq \lambda V(x) +  I_{\left( x \in C \right)}  b \quad \text{for all } x \in \sX \; .
\end{equation}
In addition, we require the associated one-step minorization condition as follows; assume that $s(x)$ is bounded below by $\varepsilon$ on $C$ and that
\begin{equation} \label{eq:eps mc}
\Pr(x,A)\ge \varepsilon Q(A)\quad \mbox{ for all } x\in C,\ \ A\in{\cal B}(\sX) \; .
\end{equation}

\citet{hobe:robe:2004} provide the following bound on $\Pr(\tau \ge n)$ based results in \cite{robe:twee:1999}.  Define $A = \sup _{x \in C} E\left[  V(X_{i+1}) | X_i = x \right]$, $J = (A - \varepsilon) / \lambda$, and
\begin{equation*}
\beta^{*} = \begin{cases} \lambda^{-1} & \text{ if } J < 1 \; , \\ \displaystyle \exp \left\{ \frac{\log \lambda \log (1 - \varepsilon) }{\log J - \log (1 - \varepsilon) } \right\} \le \lambda^{-1} & \text{ if } J \ge 1 \; . \end{cases}
\end{equation*}
Then letting $\phi(\beta) = \log \beta / \log \lambda^{-1}$, if $\beta \in (1 , \beta^{*})$, we have
\begin{align}
\Pr(\tau \ge n) & \le \beta \left[ \frac{b}{\varepsilon ( 1 - \lambda)} \right] ^ {\phi(\beta)} \left[ \frac{1 - \beta ( 1 - \varepsilon ) }{1 - ( 1 - \varepsilon ) \left( J /  ( 1 - \varepsilon ) \right) ^ {\phi(\beta)} } \right] \beta^{-n} \label{eq:tau bound} \\
& = M d(n) \; , \notag
\end{align}
where $d(n) = \beta^{-n}$ and
\[
M = \beta \left[ \frac{b}{\varepsilon ( 1 - \lambda)} \right] ^ {\phi(\beta)} \left[ \frac{1 - \beta ( 1 - \varepsilon ) }{1 - ( 1 - \varepsilon ) \left( J /  ( 1 - \varepsilon ) \right) ^ {\phi(\beta)} } \right] \; .
\]
Note $\sum_{n=1}^\infty d(n) = \sum_{n=1}^\infty \beta^{-n} = \frac{1}{\beta - 1}  = D < \infty$ since $\beta \in (1 , \beta^{*})$.  Having established the inequality in \eqref{eq:tau bound}, we next detail how to generate $B\sim {\rm Bernoulli}(ap)$ when $a>1$.

\vspace*{-3pt}
%s3 ###
\section{Bernoulli factory} \label{sec:Bernoulli}
\vspace*{-3pt}

Given a sequence $W = \{W_n\}_{n\ge 1}$ of i.i.d.\ ${\rm Bernoulli}(p)$ random variables, where $p$ is unknown, a Bernoulli factory is an algorithm that simulates a random variable $B\sim {\rm Bernoulli}(f(p))$, where $f$ is a known function.  For the exact sampling algorithm, we require a Bernoulli factory where $f(p)=ap$.  This idea arrose in \cite{asmu:glyn:thor:1992} when proposing an exact sampling algorithm for general regenerative processes.

Consider $f : S \mapsto [0,1]$, where $S \subset (0,1)$.  \citet{kean:obri:1994} show is it possible to simulate a random variable $B\sim {\rm Bernoulli}(f(p))$ for all $p \in S$ if and only if $f$ is constant, or $f$ is continuous and satisfies, for some $n \ge 1$,
\begin{equation} \label{eq:NPeq1}
\min \{ f(p) , 1 - f(p) \} \ge \min \{ p , 1-p \} ^n \quad \forall p \in S \; .
\end{equation}
While \citet{kean:obri:1994} develop the necessary and sufficient
conditions on $f$, they do not provide a detailed description of
an\vadjust{\eject}
algorithm.  \citet{nacu:pere:2005} suggest a constructive algorithm via
Bernstein polynomials for fast simulation, i.e.\ the number of input
${\rm Bernoulli}(p)$ variates needed for the algorithm has
exponentially bounded tails.  However, we find no practical
implementation since it requires dealing with sets of exponential size.
Our approach is based on the recent work of
\cite{latu:kosm:papa:robe:2011}, which avoids keeping track of large
sets by introducing a single auxiliary random variable.

%Consider a continuous function $f:[0,1]\rightarrow[0,1-\omega]$ for any fixed $\omega>0$.
The general approach, in the formulation of \cite{latu:kosm:papa:robe:2011}, is to construct two random approximations to $f(p)$, denoted $U_n$ and $L_n$, which depend on $W_1,W_2,\ldots,W_n$ and satisfy
\begin{equation}\label{eq:UnLn}
1\ge U_n=U_n(W_1,\dots,W_n)\: \ge  U_{n+1} \ge L_{n+1}\ge L_n=L_n(W_1,\dots, W_n)\ge 0\: \mbox{ a.s. }
\end{equation}
The random variables $U_n$ and $L_n$ approximate $f(p)$ in the sense that $E(U_n)\searrow f(p)$ and $E(L_n)\nearrow f(p)$ as $n\rightarrow \infty$.  The decision to continue sampling or output a zero or a one in the Bernoulli factory is made using an auxiliary ${\rm Uniform}(0,1)$ variable.

% Then one can generate a ${\rm Bernoulli}(f(p))$ variate using the following algorithm \citep{latu:kosm:papa:robe:2011}:

% \noindent\hrulefill

% \noindent {\rm Bernoulli factory algorithm:}
% \begin{enumerate}
% \item Simulate $G_0\sim {\rm Uniform}(0,1)$. Set $n=1$.
% \item Obtain $U_n$ and $L_n$ given $W_1,\dots,W_n$.
% \item If $G_0\le L_n$ set $B=1$; if $G_0\ge U_n$ set $B=0$.
% \item If $L_n<G_0<U_n$, set $n=n+1$, return to step 2.
% \item Output $B$.
% \end{enumerate}

% \vspace*{-4mm}
% \noindent\hrulefill

\begin{remark}
The almost sure monotonicity requirement in \eqref{eq:UnLn} is
typically difficult to attain and thus \cite{latu:kosm:papa:robe:2011}
relax it by using super/sub\-mar\-tingales instead.
\end{remark}

%s3.1 ###
\subsection{Modified target function}

For the rejection sampling approach to simulating $T$, we have the
ability to simulate $W$ by setting $W_i = I(\tau_i \ge n)$ for $i \ge
1$.  We require a single $B\sim {\rm Bernoulli}(a p)$, where
$a=1/[Md(n)]$ is a known constant such that $a > 0$.  The outcome $B$
determines if we accept or reject the proposed value.  For $a \in (0 ,
1 ]$ we use the simple solution in Section~\ref{sec:rejection} and for
$a \in (1 , \infty ) $ we use the Bernoulli factory.

Unfortunately, the function $f(p)=ap$ on $(0, 1/a)$ does not satisfy \eqref{eq:NPeq1} and cannot be simulated via the Bernoulli factory.  However, when restricted to $f(p)=\min\{ap, 1-\omega\}$ for $\omega >0$, such a simulation is possible.

\cite{nacu:pere:2005} and \cite{latu:kosm:papa:robe:2011} provide a detailed algorithm for $a=2$ and $0<\omega < 1/4$.  %, which can be extended to $a \in (1 , \infty ) $.
Their construction requires a minimum of 65,536 input variables (see Table~\ref{tab:BF}) before the requirement $U_n \le 1$ at \eqref{eq:UnLn} is met.  This is due to the fact that $f(p)=\min\{ 2p, 1-\omega\}$ is not differentiable and the Bernstein polynomials can approximate general Lipschitz functions at a rate of $1/\sqrt{n}$ \citep[see part (i) of Lemma 6 from][]{nacu:pere:2005}.  However, when the target function $f$ is twice differentiable, the rate increases to $1/n$ (see part (ii) of the same Lemma).

This suggests the number of ${\rm Bernoulli}(p)$ input variates required may decrease significantly by using a twice differentiable $f$.  With this in mind, we propose extending $ap$ smoothly from $\left[0, \frac{1-\omega}{a}\right]$ to $[0,1]$.  Fix $\delta < \omega$ and consider the following function
\[
F:\left[0,1-\frac{1-\omega}{a}\right]\rightarrow [0, \delta)\quad \quad F(p) = \delta\int_0^{ap/\delta} e^{-t^2} dt \; ,
\]
which is bounded by $\delta$ and twice differentiable, with $F'(p) = a \exp\{ -a^2p^2/\delta^2 \}$ and $F''(p)=-2p\frac{a^3}{\delta^2} \exp\{-a^2p^2/\delta^2\} \le 0$.  Standard calculus also gives $|F''(p)| \le a^2\frac{\sqrt{2}}{\delta\sqrt{e}}$.  ($F$ is related to the Gauss error function ($\rm erf$), though it can be simply calculated from a standard normal distribution function.)

Then define our target function $f$ as
\begin{equation} \label{eq:target f}
f(p)= \begin{cases}
ap & \displaystyle \mbox{ if } p\in\left[0, \frac{1-\omega}{a}\right)\\[9pt]
\displaystyle (1-\omega)+F\left(p-\frac{1-\omega}{a}\right) & \displaystyle \mbox{ if } p\in\left[\frac{1-\omega}{a}, 1\right]
\end{cases} \; .
\end{equation}
In other words, we have extended $ap$ such that $f$ defined at
\eqref{eq:target f} is twice differentiable with $|f''|\le C\equiv
a^2\frac{\sqrt{2}}{\delta\sqrt{e}}$.  Define $a(n,k) = f(k/n)$ and
$b(n,k)=a(n,k)+C/(2n)$ using $f$ at \eqref{eq:target f}, then we can
state Algorithm 4 of \cite{latu:kosm:papa:robe:2011} with our
modification.

\noindent\hrulefill
\vspace*{-6pt}
\begin{algorithm}\label{algI}
%\noindent {\rm \textbf{Algorithm I}:}

\mbox{}
\begin{enumerate}
\item Simulate $G_0\sim {\rm Uniform}(0,1)$.
\item Compute $m = \min \{ m \in \N : b( 2^m , 2^m ) \le 1 \}$.  Set $n=2^m$, $\tilde{L}_{\frac{n}{2}} = 0$ and $\tilde{U}_{\frac{n}{2}} = 1$.
\item Compute $H_n = \sum_{i=1}^{n} W_{i}$, $L_n = a(n, H_n)$ and $U_n = b(n, H_n)$.
\item Compute
 \begin{equation*}
L_{n}^{*} = \sum_{i=0}^{H_n} \frac{ \displaystyle {n-\frac{n}{2} \choose H_n - i} {\frac{n}{2} \choose i} }{ \displaystyle {n \choose H_n} } a\left( \frac{n}{2}, i \right) \text{ and }
U_{n}^{*} = \sum_{i=0}^{H_n} \frac{ \displaystyle {n-\frac{n}{2} \choose H_n - i} {\frac{n}{2} \choose i} }{ \displaystyle {n \choose H_n} } b\left( \frac{n}{2}, i \right) \; .
\end{equation*}
\item Compute
\begin{equation*}
\tilde{L}_{n} = \tilde{L}_{\frac{n}{2}} + \frac{L_{n} - L_{n}^{*}}{U_{n}^{*} - L_{n}^{*}} \left( \tilde{U}_{\frac{n}{2}} - \tilde{L}_{\frac{n}{2}} \right) \text{ and }
\tilde{U}_{n} = \tilde{U}_{\frac{n}{2}} - \frac{U_{n}^{*} - U_{n}}{U_{n}^{*} - L_{n}^{*}} \left( \tilde{U}_{\frac{n}{2}} - \tilde{L}_{\frac{n}{2}} \right) \; .
\end{equation*}
\item If $G_0\le \tilde{L}_{n}$ set $B=1$; if $G_0\ge \tilde{U}_{n}$ set $B=0$.
\item If $\tilde{L}_{n}<G_0<\tilde{U}_{n}$, set $n= 2 n$, return to step 3.
\item Output $B$.
\end{enumerate}
\end{algorithm}

\vspace*{-4mm}
\noindent\hrulefill

\begin{theorem} \label{th:BF}
Suppose $a$, $\delta$ and $\omega$ are constants such that $0 < a <
\infty$ and $0 < \delta < \omega < 1$.  Further suppose $f(p)$ as
defined at \eqref{eq:target f}, $a p \in [0 , 1-\omega]$ and
$W=\{W_n\}_{n\ge 1}$ are i.i.d.\ ${\rm Bernoulli}(p)$.  Then {\rm
Algorithm~\ref{algI}} outputs $B\sim {\rm Bernoulli}( ap )$. Moreover
the probability that it needs $N > n$ iterations equals $a^2 / n \delta
\sqrt{2 e}$.
\end{theorem}
\begin{proof}
See Appendix~\ref{sec:proof}.
\end{proof}
\begin{remark}
Note that the probability Algorithm~\ref{algI} needs $N > n$ iterations is
independent of the unknown value of $p$.  Hence the number of
$\mbox{Bernoulli}(p)$ variates required will also be independent of
$p$.
\end{remark}

Algorithm~\ref{algI} provides a constructive algorithm for $a \in (1 , \infty)$
and reduces the number of input variables by a factor of over 100,
which we demonstrate in the following example.  This is a critical
improvement since the Bernoulli factory accounts for most of the
computational demands of the exact sampling algorithm.\looseness=-1

%s3.2 ###
\subsection{Bernoulli factory example}

Consider generating 10,000 $\mbox{Bernoulli}(ap)$ variates for various
values of $a$ while setting $p=0.01$, $\omega=1/5$ and $\delta=1/6$.
Table~\ref{tab:BF} displays the minimum number of $\mbox{Bernoulli}(p)$
variates required, along with the observed mean and standard deviation
for the count of $\mbox{Bernoulli}(p)$ variates used to generate 10,000
$\mbox{Bernoulli}(a p)$ variates.  We can see for $a=2$ the minimum and
observed mean have been reduced substantially when comparing our target
function with the \cite{nacu:pere:2005} target (N\&P).  The reduction
in observed input\break $\mbox{Bernoulli}(p)$ variates represents a 120 times
reduction in computational time.  The N\&P implementation always stops
the simulation at the minimum, and hence the standard deviation of the
count is 0.  Table~\ref{tab:BF} also shows the input variates required
increases as $a$ increases.  Simulations for other $p$ values (not
shown) provide very similar results.

%t1 ###
\begin{table}%[b]
\centering
\caption{Comparison of count of input $\mbox{Bernoulli}(p)$ variates to implement the Bernoulli factory.}
\label{tab:BF}
\begin{tabular}{c | c | cccc }
a & 2 (N\&P) & 2  & 5 & 10 & 20 \\ \hline
Minimum & 65536 & 256 & 2048 & 8192 & 32768 \\
Mean Count & 65536 & 562.9 & 2439.8 & 10373 & 43771 \\
S.D. Count & 0 & 21046 & 7287.6 & 54836 & 3.908e5\\
\end{tabular}
\end{table}

%s4 ###
\section{Exact sampling algorithm} \label{sec:exact}

The Bernoulli factory algorithm, Theorem~\ref{th:BF}, requires $ap <
1-\omega$.  To this end, let $\kappa > 1$ such that $1/\kappa < 1-
\omega$ and from \eqref{eq:tau bound}
\[
\Pr(\tau \ge n) \le M d(n) < M d(n) \kappa \; .
\]
Then letting $a = 1 / \left[ M d(n) \kappa \right]$ we have $a \Pr(\tau
\ge n) \le 1 / \kappa < 1-\omega<1$ for all $n$.  The following
algorithm results in exact draws from $\pi$.

\noindent\hrulefill

\noindent {\rm Exact sampling algorithm for $\pi$:}
\begin{enumerate}
\item Draw $ T^* \sim \mbox{Geometric} \left( 1 - 1 / \beta \right) $, i.e.\ $\Pr \left( T^* = k \right) = \left( 1 / \beta \right)^{k-1} \left( 1 - 1 / \beta \right) $ for $k = 1, 2, \dots$, and call the result $n$.  Set $a = 1 / \left[ M d(n) \kappa \right]$.
\item If $a \le 1$, draw a single $\tau$ random variable.  Let $W = I(\tau \ge n)$ and independently draw $V \sim \mbox{Bernoulli} \left( a \right)$.  Set $B=VW$.
\item If $a > 1$, use the Bernoulli factory to obtain $B$, a single Bernoulli random variable with success probability $p = a \Pr(\tau \ge n)$.
\item If $B=1$, accept $T^*=n$; if $B=0$, return to Step 1.
\item Make a draw from $Q_n(\cdot)$.
%\vspace*{-6pt}
\end{enumerate}

\vspace*{-4mm}
\noindent\hrulefill
\eject

The algorithm requires selection of $\varepsilon$, $C$, $\lambda$, and
$\kappa$ from a range of possible values (depending on the drift and
minorization).  Further, $\beta \in (1 , \beta^{*} ) $ must be selected
depending on the previously selected parameters.  Each selection
impacts the algorithm performance and we suggest investigation of
different settings for a given example.  Our examples in
Sections~\ref{sec:toy} and~\ref{sec:real} discuss hyper-parameter
selection and provide further recommendations.

%s5 ###
\section{Toy examples} \label{sec:toy}

This section contains two toy examples in the sense that we can obtain
i.i.d.\ samples for each, hence there is no practical reason for
considering MCMC algorithms.  The purpose is to gain insights into the
exact sampler and study its performance.

%s5.1 ###
\subsection{Metropolis-Hastings example}

This section illustrates the exact sampling algorithm for a
Metropolis-Hastings sampler.  Suppose that $\sX = [0,\infty)$ and
$\pi(dx) = f_X(x) \; dx$ where $f_X(x) = e^{-x} \; I (x \ge 0) $.
Consider the function $V(x) =e^{cx}$ for some $c>0$ and suppose we use
a Metropolis sampler with a symmetric proposal density $g(\cdot|x)$,
which is supported on $[x-\gamma, x+\gamma]$, $\gamma>0$. Then, for
$x>\gamma$,
\begin{align*}
E[V(X_{i+1}) | X_i\,{=}\,x] &= \int_{x-\gamma}^xV(y)g(y|x)dy+\int_x^{x+\gamma}V(y)g(y|x)dy\frac{f(y)}{f(x)}\\[-2pt]
& + \int_x^{x+\gamma}V(x)g(y|x)dy\left(1-\frac{f(y)}{f(x)}\right) \\[-2pt]
&=\!\int_x^{x+\gamma}\!\biggl(\!V(2x\,{-}\,y)\,{+}\,V(y)\frac{f(y)}{f(x)}\,{+}\,V(x)
\biggl(\!1\,{-}\,\frac{f(y)}{f(x)}\!\biggr)\!\biggr)g(y|x)dy \\[-2pt]
&= V(x) \int_x^{x+\gamma}\!\left(\!e^{-c(y-x)}\,{+}\,e^{(c-1)(y-x)}\,{+}\,1\,{-}\,e^{-(y-x)}\!\right)\! g(y|x)dy.
\end{align*}
Selecting $g$ to be the uniform density $g(y|x)=\frac{1}{2\gamma}I_{y\in[x-\gamma, x+\gamma]}$, we get
\begin{eqnarray}
E[V(X_{i+1})|X_i=x] = V(x)\frac{1}{2\gamma}\int_0^\gamma \left( e^{-cz}+e^{(c-1)z}+1-e^{-z} \right) dz\: .
\label{r1}
\end{eqnarray}
When $x\in[0,\gamma]$,
\begin{align}
E[V(X_{i+1}) | X_i=x] &= \int_0^xV(y)g(y|x)dy+\int_x^{x+\gamma}V(y)g(y|x)dy\frac{f(y)}{f(x)}\nonumber \\[-2pt]
&+ \int_x^{x+\gamma}V(x)g(y|x)dy\left(1-\frac{f(y)}{f(x)}\right) \nonumber \\[-2pt]
&= \int_0^x e^{cy}\frac{1}{2\gamma}dy + \frac{e^{cx}}{2\gamma} \int_{x}^{x+1}e^{(c-1)(y-x)}dy \nonumber \\[-2pt]
&+ \frac{e^{cx}}{2\gamma} \int_x^{x+1} \left( 1-e^{-(y-x)} \right) dy \; .
\label{r2}
\end{align}
When combined, \eqref{r1} and \eqref{r2} obtain a drift condition for
the Metropolis-Hastings sampler considered. However, the selection of
the constants $c$ and $\gamma$ is crucial to obtaining a reasonable
computation time. The constants $\beta$ and $M$, which are described in
Section \ref{sec:bound and algorithm}, also depend on $c$ and $\gamma$.
Based on our example in Section~\ref{sec:Bernoulli}, our strategy is to
maximize $\beta$, which in turn results in small values for
$a=1/[Md(n)\kappa]$. Figure~\ref{fig:MH beta} displays a contour map of
$\beta^*$ for $c<0.3$ and $\gamma<10$. Based on this plot, we select
$c=0.028$ and $\gamma=4$, resulting in $\beta=1.0243$ (as seen below).
Evaluating the integrals in \eqref{r1} and \eqref{r2} gives the drift
condition
\begin{eqnarray*}
E[V(X_{i+1}) | X_i=x]\le \lambda V(x) + I_{\left( x \in C \right)}  b
\end{eqnarray*}
where $\lambda=0.977$, $b=0.1$ and the small set $C=[0,4]$. The
interval $[0,\gamma]$ is indeed a small set, since, when
$x\in[0,\gamma]$,
\begin{eqnarray*}
P(x,dy)&\ge& g(y|x)dy\min\left\{1, \frac{f(y)}{f(x)}\right\}=\frac{1}{2\gamma}I_{y\in[0,x+\gamma]}dy\min\left\{1, \frac{f(y)}{f(x)}\right\}\\
&\ge& \frac{1}{2\gamma}I_{y\in[0,\gamma]} e^{-y} dy\\
&=& \frac{1-e^{-\gamma}}{2\gamma}\left( \frac{1}{1-e^{-\gamma}}e^{-y}I_{y\in[0,\gamma]}dy\right) \; .
\end{eqnarray*}
This establishes the necessary minorization condition $P(x,dy)\ge s(x)Q(dy)$ where
\[
Q(dy)=q(y) dy =\frac{1}{1-e^{-\gamma}}e^{-y}I_{y\in[0,\gamma]}dy \text{ and } s(x)=\frac{1-e^{-\gamma}}{2\gamma}I_{x\in[0,\gamma]} \; .
\]
The remaining numerical elements required for the geometric probability
bound are: $A=\sup_{x\in C} E[V(X_{i+1}|X_i=x]) =  1.09197$ and $J =
(A-\varepsilon)/\lambda = 0.99283< 1$. Then $\beta =1/\lambda= 1.0243$.
Finally, the Bernoulli factory hyper-parameters are $\kappa = 5/4$ and
$\delta = 1/6$ resulting in $\omega = 0.2$.

Following \cite{mykl:tier:yu:1995}, we can now simulate the split chain
as follows: (1) draw $X_{n+1}$ from $g(\cdot|X_n)$ and accept it with
probability $\min\bigl\{1, \frac{f(X_{n+1})}{f(X_n)}\bigr\}$; (2) if
the candidate in step (1) is rejected, set $\delta_n=0$, otherwise,
generate $\delta_n$ as a Bernoulli variate with success probability
given by
 \[
 Pr(\delta_n=1)= \frac{s(X_n)q(X_{n+1})}{g(X_{n+1}|X_n)\min\left\{1, \frac{f(X_{n+1})}{f(X_n)}\right\}} \; .
\]

%f1 ###
\begin{figure}[tbp] \centering
\subfigure[Contour map of $\beta^*$.]{\includegraphics[width=2.4in]{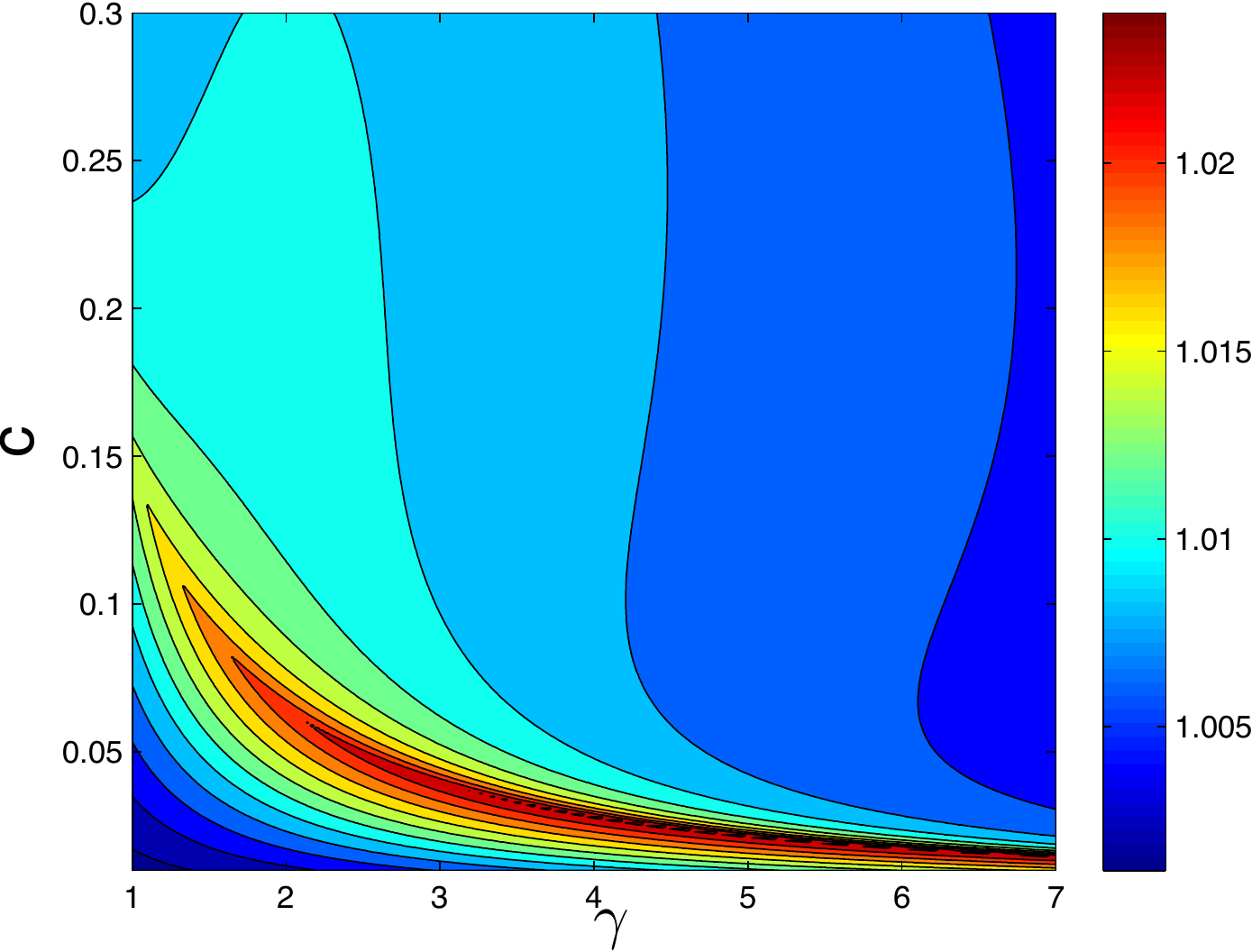} \label{fig:MH beta} }
\subfigure[Q-Q plot.]{\includegraphics[width=2.2in]{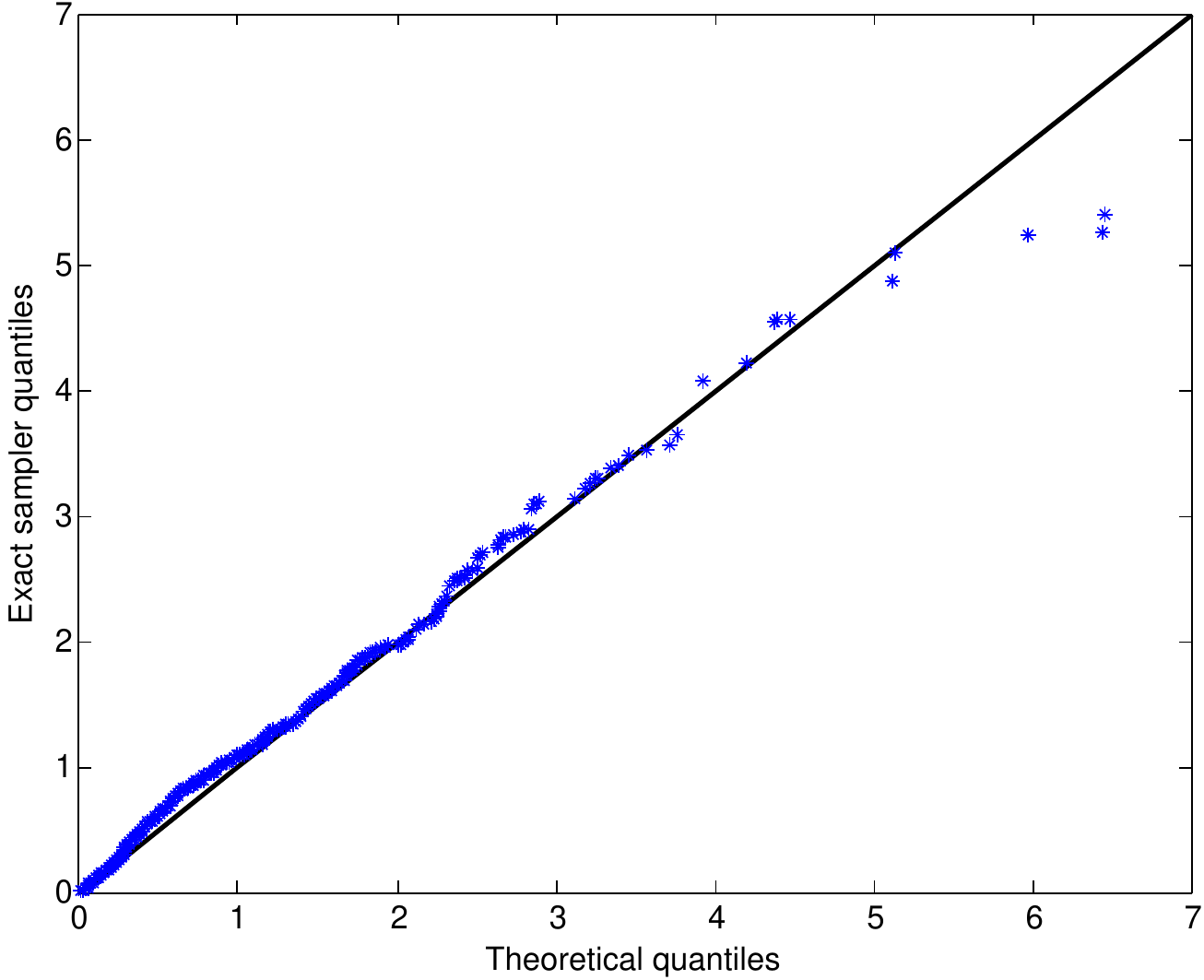} \label{fig:MH qq} }
\caption{Plots for Metropolis-Hastings example.}
\vspace*{-6pt}
\end{figure}

Using the exact sampling algorithm, we generated 1000 i.i.d.\ Exp(1)
random variates.  Figure~\ref{fig:MH qq} shows a Q-Q plot of the
observed draws versus the theoretical quantiles of the ${\rm Exp}(1)$
distribution.  During our simulations none of the proposed $T^*$ values
which required the use of the Bernoulli factory were accepted.  This is
due to the fact that our Metropolis-Hastings sampler regenerates very
fast, roughly in about 20 moves.  Thus for large $T^*$, when the
Bernoulli factory is necessary, the probability $\Pr(\tau > T^*)$ is
negligible and the Bernoulli factory outputs a zero.  Improvements via
modified drift and minorization or hyper-parameter selection may
improve this situation.

%s5.2 ###
\subsection{Gibbs example}

Suppose $Y_{i} | \mu, \theta \sim \text{N} (\mu, \theta)$ independently
for $i=1,\ldots, m$ where $m \ge 3$ and assume the standard invariant
prior $\nu(\mu, \theta) \propto \theta^{-1/2}$.  The resulting
posterior density is
\begin{equation} \label{eq:gibbs posterior}
\pi(\mu, \theta | y) \propto \theta^{-(m+1)/ 2} \exp \left\{ - \frac{m}{2\theta} (s^{2} + (\bar{y} - \mu)^{2}) \right\}
\end{equation}
where $s^{2}$ is the usual biased sample variance.  It is easy to see
the full conditional densities, $f( \mu | \theta, y)$ and $f( \theta |
\mu, y)$, are given by $\mu | \theta, y \sim \text{N}(\bar{y},
\theta/m)$ and $\theta | \mu, y \sim \text{IG} ((m-1)/2, m \left[ s^{2}
+ (\bar{y} - \mu)^{2} \right] /2)$, hence a Gibbs sampler is
appropriate.  (We say $W \sim \text{IG}(\alpha,\beta)$ if its density
is proportional to $w^{-(\alpha+1)} e^{-\beta /w} I(w > 0)$.)  We
consider the Gibbs sampler that updates $\theta$ then $\mu$; that is,
letting $x' = (\theta',\mu')$ denote the current state and $x =
(\theta,\mu)$ denote the future state, the transition looks like
$(\theta',\mu') \rightarrow (\theta,\mu') \rightarrow (\theta,\mu)$.
This Markov chain then has state space $\sX = \R ^{+} \times \R$ and
transition density
\begin{equation} \label{eq:gibbs ker}
k( \theta , \mu | \theta' , \mu') = f( \theta | \mu', y) f( \mu | \theta, y) \; .
\end{equation}
Appendix~\ref{app:Gibbs} provides a drift and minorization condition using a small set in the form
\[
C = \left\{ (\theta, \mu) \in \R^{+} \times \R : V(\mu,\theta) \le d  \right\} \: ,
\]
where $V(\theta, \mu) = 1 + \left( \mu - \bar{y} \right)^2$.

Suppose $\bar{y} = 1$, $s^{2} = 4$, and $m = 11$.  We set $\lambda =
0.5$ and $d = b / ( \lambda - (m-3)^{-1} )= 11/3$ resulting
in $\varepsilon = 0.5750034$ and $\beta^{*} = 1.3958$.  Given the
allowable range $\beta \in (1, \beta^{*} )$, we select $\beta = 1.35$,
which turns out to be extremely important for implementation.
Selection of $\beta$ very close to $\beta^*$ or close to 1 seems to
cause problems for the Bernoulli factory because of large constant
multipliers $a = 1 / \left[ M d(n) \kappa \right]$ given typical
proposed $T^{*} = n$ (which depend on $\beta$).  Experimentation has
shown us values somewhat close to $\beta^{*}$ seem to provide the best
results.  Finally, the Bernoulli factory hyper-parameters are $\kappa =
5/4$ and $\delta = 1/6$ resulting in $\omega = 0.2$.

Table~\ref{tab:gibbs constants} summarizes the resulting constants
given the hyper-parameter choices.  Notice the Bernoulli factory will
be necessary for proposed values greater than or equal to 10, that is
with probability 0.067.  The constants for values greater than 20 are
not listed.  However, these values are of interest since the minimum
number of observed $\tau$ values becomes extremely large.  While values
in this range are uncommon, $P\left( T^{*} > 20 \right) \approx 0.002$,
they occur with enough frequency to slow down the algorithm
substantially.

%t2 ###
\begin{table}[t!]
%\begin{center}
\caption{Summary constants for Gibbs sampler example with $\beta =
1.35$.  The row $\min$ refers to the minimum number of observed $\tau$s
to implement the Bernoulli factory and $p(n) =  P \left( T^{*} = n\right)$.} \label{tab:gibbs constants}
%\begin{scriptsize}
\begin{tabular*}{\textwidth}{@{\extracolsep\fill}c|cccccccccc@{}}
$ n $ & 1 &  2 &  3 &  4 &  5 &  6 &  7 &  8 &  9 & 10 \\
\hline
$p(n)$ & 0.259 & 0.192 & 0.142 & 0.105 & 0.078 & 0.058 & 0.043 & 0.032 & 0.023 & 0.017 \\
$a$ & 0.08 & 0.11 & 0.14 & 0.19 & 0.26 & 0.35 & 0.47 & 0.64 & 0.86 & 1.16 \\
$\min$ & - & - & - & - & - & - & - & - & - & 128 \\
\hline
$ n $ & 11 & 12 & 13 & 14 & 15 & 16 & 17 & 18 & 19 & 20 \\
\hline
$p(n)$ & 0.013 & 0.010 & 0.007 & 0.005 & 0.004 & 0.003 & 0.002 & 0.002 & 0.001 & 0.001 \\
$a$ & 1.57 & 2.12 & 2.87 & 3.87 & 5.22 & 7.05 & 9.52 & 12.85 & 17.35 & 23.42 \\
$\min$ & 128 & 256 & 512 & 1024 & 2048 & 4096 & 8192 & 8192 & 16384 & 32768 \\
\end{tabular*}
%\end{scriptsize}
%\end{center}
\end{table}

The exact sampling algorithm was used to generate 1000 exact draws from
the posterior density at \eqref{eq:gibbs posterior}.  Generating the
1000 draws required approximately 35 hours of computational time, about
2 minutes per draw.  A total of 27,665 $T^{*}$ values were proposed of
which 69 (0.25\%) were greater than 20.  Implementing the Bernoulli
factory required 1.52e9 $\tau$ values, or 1.52e6 $\tau$ values per
exact draw.  However, most of the computational time, and necessary
$\tau$ values, were used for a small number of proposed $T^{*}$ values.
Similar to the Metropolis-Hastings example, none of the accepted
$T^{*}$ values were from the Bernoulli factory.

The largest proposed $T^{*}$ value was 37 with $a \approx 3848$
requiring 1.07e9 $\tau$ values (of which 0 were $\ge 37$) to implement
the Bernoulli factory.  This value alone accounted for about 70\% of
the total number of $\tau$ values, and hence about 70\% of the total
computational time.  It should be noted that the largest accepted
$T^{*}$ value was 9, so proposals unlikely to be accepted account for
most of the computational demands.  Removing only the largest proposal,
the remaining 999 draws required approximately 10 hours of computation,
or about 35 seconds per draw.\looseness=-1

%f2 ###
\begin{figure}[tbp] \centering
\vspace*{-12pt}
\subfigure[Q-Q plot for $\mu$.] {\includegraphics[width=2.3in]{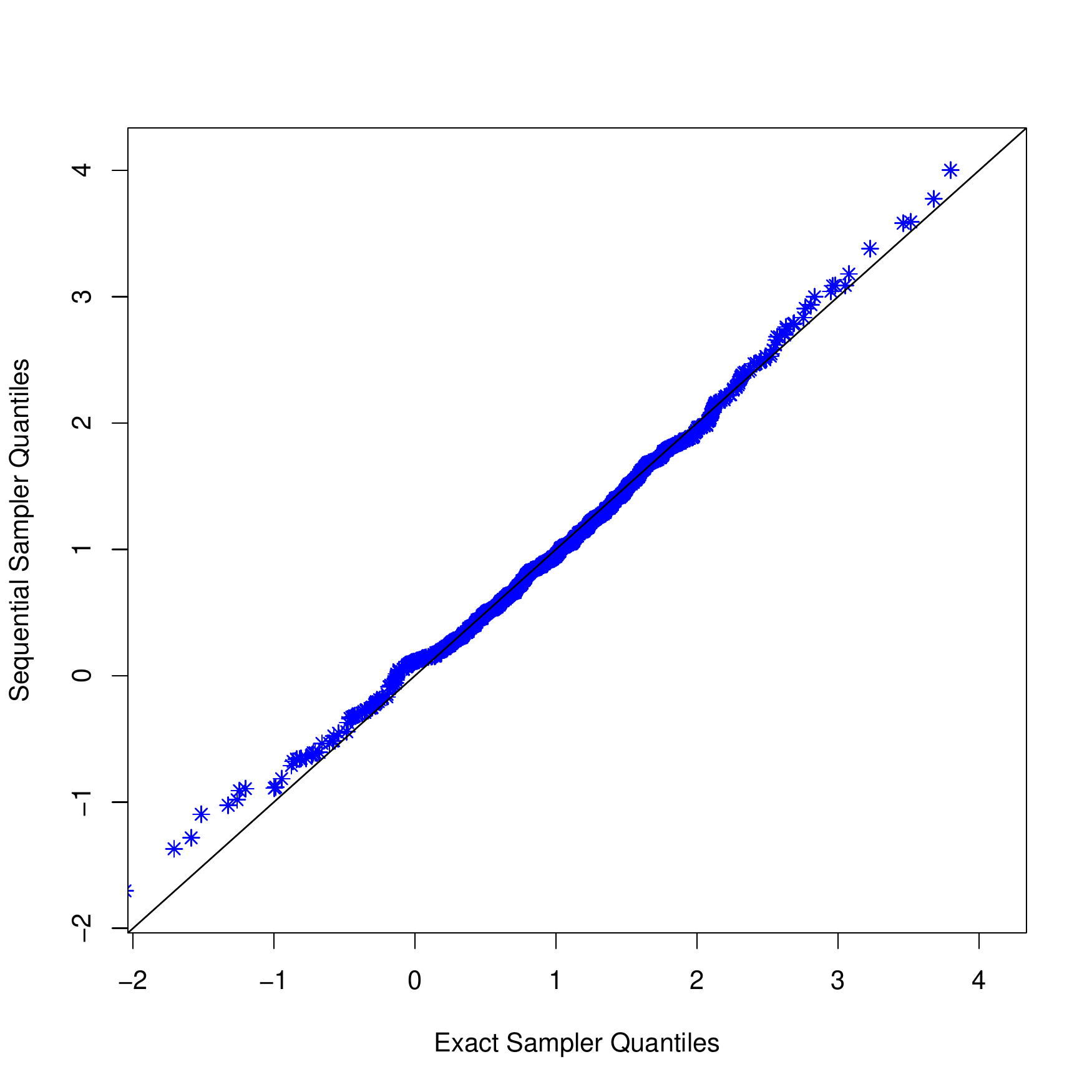} \label{fig:qq mu} }
\subfigure[Q-Q plot for $\theta$.]  {\includegraphics[width=2.3in]{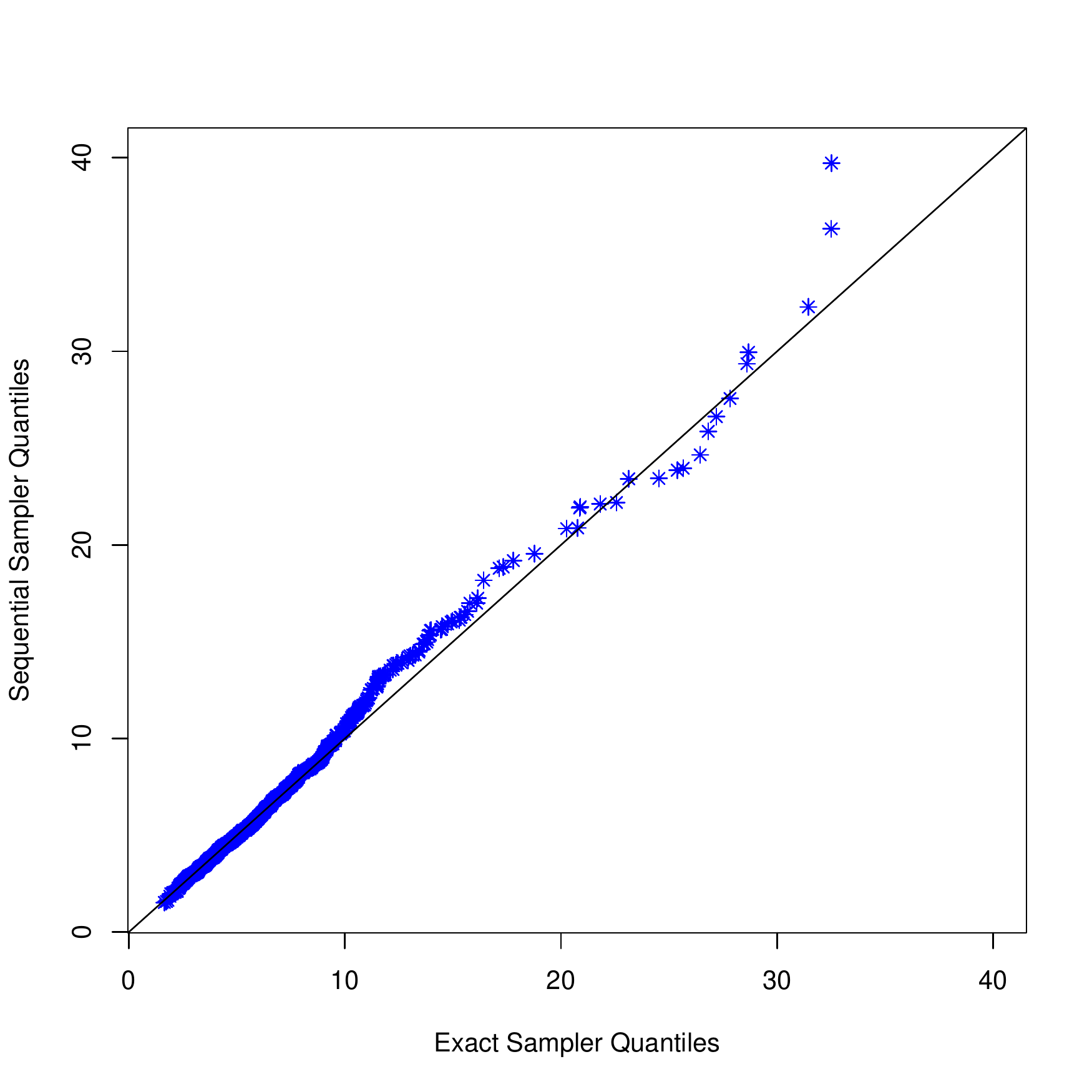} \label{fig:qq theta} }
\caption{Q-Q plots comparing 1000 exact draws to 1000 i.i.d.\ draws from a sequential sampler.}
\label{fig:gibbs qq}
\end{figure}

Alternatively for this example, we can sequentially sample from
\eqref{eq:gibbs posterior} to obtain i.i.d.\ draws
\citep{fleg:hara:jone:2008}.  Figure~\ref{fig:gibbs qq} compares the
1000 exact draws to 1000 i.i.d.\ draws from a sequential sampler using
a Q-Q plot for both $\mu$ and $\theta$.  We can see from these plots
that the exact sampling algorithm is indeed working well.\looseness=-1

%s6 ###
\section{Bayesian random effects model} \label{sec:real}

This section considers a Bayesian version of the one-way random effects model given by
\begin{equation*}
Y_{ij}=\phi_i + \zeta_{ij},\quad i=1,\ldots,q,\: j=1,\ldots, m_i
\end{equation*}
where the random effects $\phi_i$ are i.i.d.\ $N(\mu,\sigma_\phi^2)$
and independently the errors $\zeta_{ij}$ are i.i.d.\
$N(0,\sigma_e^2)$.  Thus $( \mu, \sigma_\phi^2, \sigma_e^2
)$ is the unknown parameter.

Bayesian analysis using this model requires specifying a prior
distribution, for which we consider the family of inverse gamma priors
\begin{equation*}
\pi(\mu,\sigma_\phi^2, \sigma_e^2) \propto (\sigma_\phi^2)^{-(\alpha_1+1)}e^{-\beta_2/\sigma_\phi^2} (\sigma_e^2)^{-(\alpha_2+1)}e^{-\beta_2/\sigma_e^2}
 \end{equation*}
where $\alpha_1$, $\alpha_2$, $\beta_1$ and $\beta_2$ are
hyper-parameters.  If we let $y = \{ y_{ij} \}$ and $\phi = \{ \phi_{i}
\}$ denote the vectors of observed data and random effects
respectively, then the posterior density is as follows
 \begin{equation} \label{eq:sty posterior}
 \pi \left( \phi, \mu, \sigma_\phi^2, \sigma_e^2 \right) \propto f \left( y | \phi, \mu, \sigma_\phi^2, \sigma_e^2 \right) f \left( \phi | \mu, \sigma_\phi^2, \sigma_e^2 \right) \pi(\mu,\sigma_\phi^2, \sigma_e^2) \; ,
 \end{equation}
where
\[
f \left( y | \phi, \mu, \sigma_\phi^2, \sigma_e^2 \right) = \prod_{i=1}^{q} \prod_{j=1}^{m_i} \left( 2 \pi \sigma_e^2 \right) ^{-\frac{1}{2}} \exp \left\{ - \frac{1}{2 \sigma_e^2} \left( y_{ij} - \phi_{i} \right)^2 \right\}
\]
and
\[
f \left( \phi | \mu, \sigma_\phi^2, \sigma_e^2 \right) = \prod_{i=1}^{q} \left( 2 \pi \sigma_\phi^2 \right) ^{-\frac{1}{2}} \exp \left\{ - \frac{1}{2 \sigma_\phi^2} \left( \phi_{i} - \mu \right)^2 \right\} \; .
\]
For ease of exposition, we will suppress the dependency on the data $y$
and define the usual summary statistics: $\bar y_i = m_i ^{-1} \sum_j
y_{ij}$, $M = \sum_i m_i$, $\bar{y} = M^{-1} \sum_{i} \sum_{j} y_{ij}
$, $\text{SST} =  \sum_{i} m_i \left( \bar{y}_i - \bar{y} \right)^2 $
and $\text{SSE} = \sum_{i} \sum_{j} \left( y_{ij} - \bar{y}_i
\right)^2$.

We consider a block Gibbs sampler  that updates $\theta=(\sigma_\phi^2,
\sigma_e^2)$ then $\xi=(\mu, \phi)$, that is $(\theta', \xi')
\rightarrow (\theta, \xi') \rightarrow (\theta, \xi)$.  The necessary
full conditionals can be obtained via manipulation of \eqref{eq:sty
posterior}.  That is, $f (\theta | \xi')$ is the product of two inverse
gammas such that
\begin{equation*}
\sigma^2_\phi | \xi' \sim IG\left( \frac{q}{2}+\alpha_1, \frac{w_1(\xi')}{2}+\beta_1\right)
\end{equation*}
and
\begin{equation*}
\sigma^2_e | \xi' \sim IG\left( \frac{M}{2}+\alpha_2, \frac{w_2(\xi')+SSE}{2}+\beta_2\right)
\end{equation*}
where $w_1(\xi)=\sum_{i=1}^q(\phi_i-\mu)^2$ and $w_2(\xi) = \sum m_i
(\phi_i - \bar y_i)^2$.  Further, $f(\xi | \theta)$ is multivariate
normal density whose parameters are given in \cite{tan:hobe:2009}.
This Markov chain then has state space $\sX=\R^{q+1}\times \R^2_{+}$
and transition density
\[
k(\xi, \theta | \: \xi', \theta') = f (\theta | \xi') f(\xi | \theta) = f(\sigma_\phi^2 | \xi') f(\sigma_e^2 | \xi') f(\xi | \theta) \; .
\]

Implementation of the exact sampling algorithm requires a drift and
associated minorization condition as at \eqref{eq:drift} and
\eqref{eq:eps mc}.  \cite{hobe:geye:1998}, \cite{jone:hobe:2001,
jone:hobe:2004} and \cite{tan:hobe:2009} analyze variations of the
proposed block Gibbs sampler, however none obtain sufficient constants
for a practical implementation of our algorithm.  To this end, the
following theorem improves upon the drift constants of
\cite{tan:hobe:2009} for a balanced design while using a simplified
version of their drift function.

\begin{theorem} \label{th:drift}
Let $m_i = m$ for all $i = 1, \dots, q$ and let $\Delta_2 =1-1/[q(m+1)]
+ \max\{ q(m+1)/m^2\:, 1/m \}$.  Further let $V:\R^{q+1}\times \R^2_{+}
\rightarrow [1,\infty)$ such that $V(\xi,\theta) = K + \delta_1
w_1(\xi) + \delta_2w_2(\xi)$ where $K \ge 1$, $\delta_1 > 0$ and
$\delta_2 >0$ and define
\[
\lambda^* = \max\left\{  \frac{1}{q+2\alpha_1-2}\:, \frac{\delta_1\Delta_2/\delta_2 +q+1}{M+2\alpha_2-2} \right\} \; .
\]
Then there exists $K \ge 1$, $\delta_1 > 0$ and $\delta_2>0$ such that
$\lambda^* < 1$ and \eqref{eq:drift} holds.  That is, for any
$\lambda\in(\lambda^*, 1)$,
\begin{equation}
E(V(\xi,\theta) | \xi', \theta') \le \lambda V(\xi', \theta') + b 1_{(\xi', \theta')\in C}
\label{eq:sty drift}
\end{equation}
where
\begin{align*}
b &= K(1-\lambda) +  \frac{2\delta_1\beta_1}{q+2\alpha_1-2} + \frac{(\delta_1\Delta_2+\delta_2(q+1))(SSE+2\beta_2)}{M+2\alpha_2-2} \\
&+(\delta_1+m\delta_2) \sum_{i=1}^q(\bar y_i - \bar y)^2
\end{align*}
and
\[
C=\{ (\xi, \theta) \in \R^{q+1}\times\R_+^2\: : \: V(\xi,\theta) \le d\}
\]
where $ d = b/(\lambda-\lambda^*)$.
\end{theorem}
\begin{proof}
See Appendix~\ref{ap:sty}.
\end{proof}

The drift condition still holds if we increase the small set to
\[
C=\{ (\xi, \theta) \in \R^{q+1}\times\R_+^2\: : \: K+\delta_1 w_1(\xi) \le d, K+\delta_2 w_2(\xi)\le d\} \; ,
\]
for which Appendix~\ref{ap:sty} provides the associated minorization condition.

%s6.1 ###
\subsection{Styrene exposure dataset}

We will implement the exact sampling algorithm using the styrene
exposure dataset from \cite{lyle:kupp:rapp:1997} analyzed previously by
\cite{jone:hobe:2001} and \cite{tan:hobe:2009}.  The data, summarized
in Table~\ref{tab:sty}, is from a balanced design such that $m_i=m = 3$
for $i=1,\ldots, q$, $q = 13$, and $M=mq = 39$.

%t3 ###
\begin{table}[b]
\vspace*{-6pt}
\centering
\tabcolsep=12.5pt
\caption{Styrene exposure data summary statistics.}
\label{tab:sty}
\begin{tabular}{@{\ } c|ccccccc@{\ }}
Worker & 1 & 2 & 3 & 4 & 5 & 6 & 7 \\
$\bar{y}_i$ & 3.302 & 4.587 & 5.052 & 5.089 & 4.498 & 5.186 & 4.915 \\
\hline
Worker & 8 & 9 & 10 & 11 & 12 & 13 & \\
$\bar{y}_i$ & 4.876 & 5.262 & 5.009 & 5.602 & 4.336 & 4.813 & \\
\hline
& \multicolumn{7}{c}{$\bar{y} = 4.089, \quad \text{SST} = 11.430, \quad \text{SSE} = 14.711 $}
% & \multicolumn{7}{c}{$\bar{y} = M^{-1} \sum_{i=1}^{13} \sum_{j=1}^{3} y_{ij} = 4.089 $} \\
% & \multicolumn{7}{c}{$\text{SST} = 3 \sum_{i=1}^{13} \left( \bar{y}_i - \bar{y} \right)^2 = 11.430 $} \\
% & \multicolumn{7}{c}{$\text{SSE} = \sum_{i=1}^{13} \sum_{j=1}^{3} \left( y_{ij} - \bar{y}_i \right)^2 = 14.711 $} \\
\end{tabular}
\end{table}

We consider prior hyper-parameter values $\alpha_1 = \alpha_2 = 0.1$
and $\beta_1=\beta_2 = 10$.  The drift function at \eqref{eq:sty drift}
requires specification of drift parameters $K=50$, $\Delta_2=6.7585$,
$\delta_1 = 1$, and $\delta_2 = 1$ which results in $\lambda^*=0.5580$.
We then choose $ \lambda=0.97$, $b=37.88927$ and $d=91.96992$,
resulting in $\varepsilon=0.01269784$.

%t4 ###
\begin{table}[!b]
\vspace*{-6pt}
\tabcolsep=12pt
%\begin{center}
\centering
\caption{List of 20 i.i.d.\ $\theta$ draws from the posterior at \eqref{eq:sty posterior} with the accepted $T^*$ values.}
\label{tab:20draws}
\begin{tabular}{@{\ }c|ccc@{\ }}
Draw & $T^*$ & $\sigma^2_{\phi}$ & $\sigma^2_{e}$ \\
  \hline
1 & 145 & 2.477 & 1.2624 \\
  2 & 18 & 1.234 & 1.7698 \\
  3 & 286 & 3.058 & 1.8791 \\
  4 & 40 & 2.607 & 1.2079 \\
  5 & 76 & 2.177 & 2.0603 \\
  6 & 287 & 6.513 & 1.2870 \\
  7 & 39 & 5.961 & 1.5295 \\
  8 & 103 & 1.642 & 1.2093 \\
  9 & 194 & 2.150 & 1.8129 \\
  10 & 195 & 2.112 & 1.4871 \\
  11 & 101 & 1.101 & 1.7166 \\
  12 & 2 & 1.659 & 1.0805 \\
  13 & 5 & 5.544 & 1.2856 \\
  14 & 9 & 1.505 & 1.6600 \\
  15 & 1 & 2.105 & 1.4137 \\
  16 & 150 & 2.681 & 0.7317 \\
  17 & 63 & 3.131 & 1.1506 \\
  18 & 64 & 3.514 & 1.8245 \\
  19 & 52 & 2.119 & 1.1743 \\
  20 & 62 & 3.571 & 1.5009 \\
\end{tabular}
\vspace*{-6pt}
%\end{center}
\end{table}

Using these settings $\beta^{*} = 1.000092$ and we choose $\beta =
1.000083$ resulting in $M = 10.19413$.  We again used Bernoulli factory
hyper-parameters of $\kappa = 5/4$ and $\delta = 1/6$ resulting in
$\omega = 0.2$.  In this case, the Bernoulli factory is necessary for
proposals greater than 30,706, approximately 8\% of proposed values.
It is extremely likely the output from the Bernoulli factory will be
zero since a sample of 1000 i.i.d.\ $\tau$ values yielded only a
maximum of 1278.

The exact sampling algorithm was run until we obtained a 20 i.i.d.\
draws from the posterior at \eqref{eq:sty posterior} which took 31,887
proposed $T^*$ values and 2.61e8 $\tau$ values for the Bernoulli
factory.  The accepted $T^*$ values and i.i.d.\ $\theta$ values are
listed in Table~\ref{tab:20draws}.  Notice the maximum accepted $T^*$
was 287, which is well below our observed maximum of 1278 from 1000
i.i.d.\ $\tau$ values.  Hence, drawing from $Q_n$ was easy and almost
all of the simulation time was used for the Bernoulli factory.
Obtaining the 20 i.i.d.\ draws required 24 days of computational time
utilizing six processors in parallel (equating to 144 days on a single
processor).

%s7 ###
\section{Discussion} \label{sec:dis}

This paper describes an exact sampling algorithm using a geometrically
ergodic Markov chain on a general state space.  The algorithm is
applicable for any Markov chain where one can establish a drift and
associated minorization with computable constants.  The limitation of
the method is that the simulation time may be prohibitive.

\citet{blan:thom:2007} implement an approximate version using the
Bayesian probit regression example from \citet{vand:meng:2001} with
regeneration settings provided by \citet{roy:hobe:2007}.  This example
is ill-suited using the proposed algorithm because of computational
limitations related to the Bernoulli factory and in obtaining a
practical $\varepsilon$.  Specifically, we found (in simpler examples)
obtaining a single draw from $\pi$ sometimes required millions of
i.i.d.\ $\tau$ variates.  Unfortunately, even using non-constant
$s(x)$, the probit example requires about 14,000 Markov chain draws per
$\tau$ \citep{fleg:jone:2010}.  Hence obtaining a single draw from
$\pi$ would require an obscene number of draws from $X$.
Implementation for more complicated Markov chains, such as this, likely
requires further improvements, or a lot of patience.

Careful analysis of the Markov chain sampler is necessary to find
useful drift and minorization constants.  Most research establishing
drift and minorization is undertaken to prove geometric ergodicity, in
which case the obtained constants are of secondary importance.
However, performance of the exact sampling algorithm is heavily
dependent on these constants.  Improving them may be enough to obtain
exact samples in many settings.

Alternatively, the speed of the overall algorithm would improve if one
could find a bound using non-constant $s(x)$ or a sharper bound with
$\varepsilon$.  The current bound at \eqref{eq:tau bound} could
potentially be modified upon by only considering specific models, or
specific classes of models.

Finally, one could obtain further improvements to the Bernoulli factory
since it requires most of the necessary $\tau$ variates.  Our work has
already obtained a 100 times reduction in computational time.  However
there may be further improvements available for the Bernstein
polynomial coefficients, modifications to Algorithm 4 of
\cite{latu:kosm:papa:robe:2011} or an entirely different method to
estimate $f$.  Hyper-parameter settings also impact performance and
could be investigated further.

\section*{Acknowledgments}

This work was started during the AMS/NSF sponsored Mathematics Research
Communities conference on ``Modern Markov Chains and their Statistical
Applications,'' which was held at the Snowbird Ski and Summer Resort in
Snowbird, Utah between June 28 and July 2, 2009.  We are especially
grateful to Jim Hobert who proposed this research problem.  We would
also like to thank fellow conference participants Adam Guetz, Xia Hua,
Wai Liu, Yufei Liu, and Vivek Roy.  Finally, we are grateful to the
anonymous referees, anonymous associate editor, Galin Jones, Ioannis
Kosmidis and Krzysztof Latuszy{\' n}ski for their constructive comments
and helpful discussions, which resulted in many improvements.

\begin{appendix}

% \section*{Appendices}

%s8 ###
\section{Proof of Theorem~\ref{th:BF}} \label{sec:proof}

\begin{proof}
By construction $f$ is a smooth function
$f:[0,1]\rightarrow[0,1-\varepsilon]$ for some $0 < \varepsilon <
\omega - \delta$.  Proposition 3.1 of \cite{latu:kosm:papa:robe:2011}
and Lemma 6 of \citet{nacu:pere:2005} prove existence of an algorithm
that simulates $f(p)$ if
\begin{itemize}
\item[(i)] $f$ has second derivative $f''$ which is continuous and
\item[(ii)] the coefficients $a$ and $b$ satisfy
\beqn
a(2n,k){2n\choose k}&\ge&\sum_{i=0}^ka(n,i){n\choose i}{n \choose {k-i}}\; ,
\label{a1}\\
%\eeqn
%\vspace*{-18pt}
%\beqn
b(2n,k){2n\choose k}&\le&\sum_{i=0}^kb(n,i){n\choose i}{n \choose {k-i}}\; .
\label{b1}
\eeqn
\end{itemize}

Condition (i) is clearly satisfied by construction, so it remains to
check condition (ii).  Since the coefficients $a$ and $b$ are defined
through $f$, inequalities \eqref{a1} and \eqref{b1} will be checked
using the properties of $f$.

Recall $b(n,k)=a(n,k)+C/(2n)$, the inequalities \eqref{a1} and
\eqref{b1} above can be re-expressed \citep{nacu:pere:2005} as
\[
a(2n,k)\ge E(a(n,X)) \text{ and } b(2n,k)\le E(b(n,X)),
\]
where $X$ is a hypergeometric random variable, with parameters
$(2n,k,n)$. Using the definition of $a$, and the fact that $f$ is
concave, the first inequality is a direct application of Jensen's
inequality.  The second part is a straight forward application of Lemma
6 from \citet{nacu:pere:2005} and the properties of the hypergeometric
distribution.

Finally, the probability the algorithm needs $N > n$ follows directly
from definitions of the coefficients $a$ and $b$ and Theorem 2.5 of
\cite{latu:kosm:papa:robe:2011}.
\end{proof}

%s9 ###
\section{Toy Gibbs drift and minorization} \label{app:Gibbs}

%s9.1 ###
\subsection{Drift condition}

Let $X=\{X_n\}_{n\ge 0}$ be the Markov chain corresponding to the Gibbs
transition kernel given in \eqref{eq:gibbs ker}.  Recall $\sX = \R ^{+}
\times \R$, $x' = (\theta',\mu')$ denotes the current state and $x =
(\theta,\mu)$ denotes the future state.  \citet{jone:hobe:2001}
establish a minorization and \citet{rose:1995a}--type drift condition
for $m \ge 5$, and hence prove the associated Markov chain is
geometrically ergodic.  Using their argument we show a
Roberts-and-Tweedie-type drift condition at \eqref{eq:drift}
\citep{robe:twee:1999, robe:twee:2001} holds using the function $V(x) =
V(\theta, \mu) = 1 + \left( \mu - \bar{y} \right)^2$.  Conditional
independence \citep[based on the update order, see][]{jone:hobe:2001}
yields
\begin{align*}
E\left[  V( X_{i+1} ) | X_i=x' \right] & =  E\left[  V( \theta, \mu ) | \theta' , \mu' \right] \\
& =  E\left[  V( \theta, \mu ) | \mu' \right] \\
& = E \left\{ E \left[  V( \theta, \mu ) | \theta \right] | \mu' \right\} \; .
\end{align*}
Since $\mu | \theta, y \sim \text{N}(\bar{y}, \theta/m)$, the inner expectation is
\begin{equation*}
E \left[  V( \theta, \mu ) | \theta \right] = E \left[  \left( 1 + \left( \mu - \bar{y} \right)^2 \right) | \theta \right] = 1 + \Var \left( \mu | \theta \right) = 1 + \frac{\theta}{m} \; .
\end{equation*}
Then since $\theta | \mu', y \sim \text{IG} ((m-1)/2, m \left[ s^{2} + (\bar{y} - \mu')^{2} \right] /2)$,
\begin{equation*}
E \left[ \theta | \mu' \right] = \frac{m \left[ s^{2} + (\bar{y} - \mu')^{2} \right]}{m-3} \; ,
\end{equation*}
and hence
\begin{equation} \label{eq:expV}
E\left[  V( X_{i+1} ) | X_i=x' \right] = \frac{1 + \left( \mu' - \bar{y} \right)^2}{m-3} + \frac{s^2 + m - 4}{m-3} \; .
\end{equation}
Let $\lambda \in \left( (m-3)^{-1}, 1 \right)$, $b = \left( s^2 + m -4 \right) / (m-3)$, $d \ge b / \left( \lambda - (m-3)^{-1} \right)$ and
\[
C = \left\{ (\theta, \mu) \in \R^{+} \times \R : V(\mu,\theta) \le d  \right\} \: ,
\]
then the drift condition at \eqref{eq:drift} is satisfied, that is
\begin{equation*}
E\left[  V( X_{i+1} ) | X_i=x' \right] \le \lambda V(x') + I_{\left( x' \in C \right)}  b \text{ for all } x ' \in \sX \; .
\end{equation*}
It is easy to see from \eqref{eq:expV} that
\[
A=\sup_{x' \in C} E\left[  V(X_{i+1}) | X_i = x \right] = \frac{ \sup_{x' \in C} V(x') }{m-3} + b = \frac{d}{m-3} + b \; .
\]

%s9.2 ###
\subsection{Minorization condition}

Now we establish the associated minorization condition at \eqref{eq:eps
mc} using a similar argument to \citet{jone:hobe:2001}.  Let
$C_\mu=\{\mu\in\R\: : \: 1+(\mu-\bar y)^2\le d\}$, then for any $\mu'
\in C_\mu$
\begin{equation*}
k( \theta , \mu | \theta' , \mu') = f(\theta | \mu')f(\mu|\theta) \ge f(\mu | \theta) \inf_{\mu \in C_\mu} f(\theta | \mu) \; .
\end{equation*}
Recall $f(\theta | \mu)$ is an IG density, thus $g(\theta) := \inf_{\mu \in C_\mu} f(\theta | \mu)$ can be written in closed form \citep{rose:1996, jone:hobe:2004, tan:hobe:2009},
\begin{eqnarray*}
\inf_{\mu \in C_\mu}f(\theta | \mu) = \left \{
\begin{array}{lr}
IG\left( \frac{m-1}{2}, \frac{m(s^2+d-1)}{2}; \theta\right) & \mbox{ if } \theta<\theta^*\\[6pt]
IG\left( \frac{m-1}{2}, \frac{ms^2}{2}; \theta \right) & \mbox{ if } \theta\ge\theta^*\\
\end{array}
\right.
\end{eqnarray*}
where
$\theta^*=m(d-1)\left[(m-1)\log\left(1+\frac{d-1}{s^2}\right)\right]^{-1}$
and $IG(\alpha,\beta;x)$ is the inverse gamma density evaluated at $x$.
If we further define
\begin{equation*}
\varepsilon = \int_{\R\times\R_{+}}f(\mu | \theta)\inf_{\mu \in C_\mu}f(\theta | \mu) d\mu\: d\theta =\int_{\R_+}\inf_{\mu \in C_\mu}f(\theta | \mu) d\theta
\end{equation*}
and density $q(\theta, \mu) = \varepsilon^{-1} g(\theta) f(\mu | \theta)$, then
\begin{equation*}
k( \theta , \mu | \theta' , \mu') \ge \varepsilon q(\theta, \mu) \; .
\end{equation*}
Letting $Q(\cdot)$ be the probability measure associated with the
density $q$, then the minorization condition from \eqref{eq:eps mc}
holds, that is for any set $A$ and any $(\theta', \mu') \in C$
\begin{equation*}
 P(x,A) \ge \varepsilon \, Q(A) \text{ for all } A \in {\cal B}(\sX) \;.
\end{equation*}
Notice the minorization condition holds for any $d > 0$.

Simulating the split chain requires evaluation of \eqref{eq:delta equals 1},
\begin{align*}
Pr \left( \delta' = 1 | \mu', \theta', \mu, \theta \right) & = \frac{ \varepsilon \; q ( \theta, \mu )}{ k \left( \mu , \theta | \mu' , \theta' \right) } \\
& = \frac{ \varepsilon \; \varepsilon^{-1} \; g(\theta) \; f ( \mu | \theta ) }{ f \left( \theta | \mu' \right) \;  f \left( \mu | \theta \right) } \\
& = \frac{ g(\theta) }{ f \left( \theta | \mu' \right) } \\
& = \frac{ \left[ \frac{ m ( s^2 +  I_{\{ \theta < \theta^{*} \} } (d-1) ) }{2} \right]^{(m-1)/2} }
{ \left[ \frac{ m ( s^2 +  (\bar{y} - \mu')^2 ) }{2} \right]^{(m-1)/2} }\\
& \cdot \frac{ \frac{\theta ^{-(m-1)/2 -1 } }{\Gamma \left( (m-1)/2 \right)}
\exp \left \{ - \left[ \frac{ m ( s^2 + I_{\{ \theta < \theta^{*} \} } (d-1) ) }{2 \theta } \right] \right \} }
{ \frac{\theta ^{-(m-1)/2 -1 } }{\Gamma \left( (m-1)/2 \right)}
\exp \left \{ - \left[ \frac{ m ( s^2 + (\bar{y} - \mu')^2 ) }{2 \theta } \right] \right \} }\\
& = \left[ \frac{ s^2 +  I_{\{ \theta < \theta^{*} \} } (d-1) }{ s^2 + (\bar{y} - \mu')^2 } \right]^{(m-1)/2} \\
& \cdot \exp \left \{ - \frac{ m I_{\{ \theta < \theta^{*} \} } (d-1) }{2 \theta } + \frac{ m  (\bar{y} - \mu')^2 }{2 \theta }  \right \} \; .
\end{align*}
Notice that $Pr \left( \delta' | \mu', \theta' , \mu , \theta \right)$ is free of $\varepsilon$, $\theta'$, and $\mu$.

%s10 ###
\section{One-way random effects drift and minorization} \label{ap:sty}

%s10.1 ###
\subsection{Proof of Theorem~\ref{th:drift}}
\begin{proof}
Notice
\begin{equation}
E(V(\xi,\theta) | \xi', \theta') = E\{  E( V(\xi,\theta) | \theta  ) | \xi', \theta'\}
\label{sty:d1}
\end{equation}
where the inner expectation becomes
\[
E(V(\xi,\theta) | \theta) = K+\delta_1 E(w_1(\xi) | \theta) + \delta_2 E(w_2(\xi) | \theta) \; .
\]
Note that
\begin{eqnarray*}
E(w_1(\xi) | \theta) =  E\Big[ \sum_{i=1}^q (\phi_i - \mu)^2 | \theta\Big] = \sum_{i=1}^q \left\{{\rm Var}( \phi_i-\mu |\theta) + \Big( E(\phi_i-\mu | \theta)\Big)^2\right\} \; .
\end{eqnarray*}
For general designs (balanced or unbalanced) \cite{tan:hobe:2009} prove
\begin{eqnarray*}
 \sum_{i=1}^q {\rm Var}(\phi_i-\mu |\theta) \le \Delta_1\sigma_\phi^2 +\Delta_2\sigma_e^2
\end{eqnarray*}
where,
\begin{equation*}
\Delta_1 = \min\left\{ q\left(\sum_{i=1}^q \frac{m_i}{m_i+1}\right)^{-1}\: , \frac{q \cdot \max\{m_1, \ldots, m_q\} }{M} \right\}
\end{equation*}
and
\begin{equation*}
\Delta_2 = \sum_{i=1}^q\frac{1}{m_i} - \sum_{i=1}^q\frac{1}{M(1+m_i)} + \max\left\{  q\left(\sum_{i=1}^q \frac{m_i}{m_i+1}\right)^{-1}  \: , \frac{q}{M}  \right\} \; .
\end{equation*}
Simplifying $\Delta_1$ and $\Delta_2$ under the balanced design, we
obtain $\Delta_1 = 1$ and $\Delta_2 =1-1/[q(m+1)] + \max\{
q(m+1)/m^2\:, 1/m \}$.

From the full conditionals \citep{tan:hobe:2009}
\begin{eqnarray*}
E(\phi_i-\mu | \theta ) = \frac{m\sigma_\phi^2}{\sigma_e^2+m\sigma_\phi^2}(\bar y_k - \bar y) \; ,
 \end{eqnarray*}
then it follows that (with $\Delta_1 = 1$)
\begin{eqnarray*}
E(w_1(\xi) | \theta )&\le& \sigma_\phi^2 +\Delta_2\sigma_e^2  + \left(\frac{m\sigma_\phi^2}{\sigma_e^2+m\sigma_\phi^2}\right)^2\sum_{i=1}^q(\bar y_i - \bar y)^2\\
&\le&\sigma_\phi^2 +\Delta_2\sigma_e^2  + \sum_{i=1}^q(\bar y_i - \bar y)^2 \; .
\end{eqnarray*}
Similarly,
 \begin{eqnarray*}
E(w_2(\xi) | \theta)&\le& (q+1)\sigma_e^2 + m \sum_{i=1}^q(\bar y_i - \bar y)^2
\end{eqnarray*}
To complete the calculation in \eqref{sty:d1}, recall $\sigma_\phi^2 | \xi'$ and $\sigma_e^2 | \xi'$, thus
\[
E(\sigma_\phi^2 | \xi') = \frac{w_1(\xi')+2\beta_1}{q+2\alpha_1-2}
\]
and
\[
E(\sigma_e^2 | \xi') = \frac{w_2(\xi')+SSE+2\beta_2}{M+2\alpha_2-2} \; .
\]
It follows that
\begin{align}
E(V(\xi,\theta) | \xi', \theta') &\le K+ \frac{ 1}{q+2\alpha_1-2}\Big(\delta_1 w_1(\xi')\Big) + \frac{\delta_1\Delta_2/\delta_2 +q+1}{M+2\alpha_2-2}\Big(\delta_2 w_2(\xi')\Big) \notag \\
& + \frac{2\delta_1\beta_1}{q+2\alpha_1-2} + \frac{(\delta_1\Delta_2+\delta_2(q+1))(SSE+2\beta_2)}{M+2\alpha_2-2} + \notag \\
& + (\delta_1+m\delta_2) \sum_{i=1}^q(\bar y_i - \bar y)^2\; . \label{eq:sty EV}
\end{align}
Note there exists $\delta_1 > 0 $ and $\delta_2>0$ such that
$\lambda^*<1$, hence the drift condition at \eqref{eq:drift} is
satisfied.
\end{proof}

For the exact sampling algorithm, we can see from \eqref{eq:sty EV} and the definitions of $C$ and $b$ that
\[
A=\sup_{x' \in C} E\left[  V(X_{i+1}) | X_i = x \right] \le \frac{ (d-K) }{q+2\alpha_1-2}  + \frac{\delta_1\Delta_2/\delta_2 +q+1}{M+2\alpha_2-2} \left( d-K \right)  + b + \lambda^{*} \; .
\]

%s10.2 ###
\subsection{Minorization condition}

Next we show the associated minorization condition holds.  The argument
will be similar to the toy Gibbs example from Appendix~\ref{app:Gibbs}.
Let $C_\xi=\{ \xi \in \R^{q+1}\: : \: 1+\delta_1 w_1(\xi) \le d,
1+\delta_2 w_2(\xi)\le d \}$, then the associated minorization
condition holds if we can find an $\varepsilon$ and $q(\xi,\theta)$
such that for any $\xi' \in C_\xi$
\begin{align}
k(\xi, \theta | \: \xi', \theta') &= f(\sigma_\phi^2 | \xi') f(\sigma_e^2 | \xi') f(\xi | \theta) \notag \\
&\ge f(\xi | \theta) \inf_{\xi \in C_\xi}f(\sigma_\phi^2 | \xi )\inf_{\xi \in C_\xi}f(\sigma_e^2 | \xi ) \notag \\
&= f(\xi | \theta) g_1(\sigma_\phi^2) g_2(\sigma_e^2) \notag \\
&= \varepsilon q(\xi,\theta) \label{eq:sty minor} \; .
\end{align}
The two infimums can be found analytically as before:
\begin{eqnarray*}
g_1(\sigma_\phi^2) = \inf_{\xi \in C_\xi}f(\sigma_\phi^2 | \xi ) = \left\{
\begin{array}{ll}
IG\left( \frac{q}{2}+\alpha_1, \frac{d-K}{2 \delta_1}+\beta_1; \sigma_\phi^2 \right) & \mbox{ if } \sigma_\phi^2 \le \sigma_\phi^*\\[6pt]
 IG\left( \frac{q}{2}+\alpha_1, \beta_1; \sigma_\phi^2 \right) & \mbox{ if } \sigma_\phi^2 > \sigma_\phi^*
 \end{array}
 \right.
 \end{eqnarray*}
and
\begin{eqnarray*}
g_2(\sigma_e^2) = \inf_{\xi \in C_\xi}f(\sigma_e^2 | \xi ) = \left\{
\begin{array}{ll}
IG\left(  \frac{M}{2}+\alpha_2, \frac{(d-K)/\delta_2+SSE}{2}+\beta_2\right) & \mbox{ if } \sigma_e^2\le\sigma_e^*\\[6pt]
 IG\left( \frac{M}{2}+\alpha_2, \frac{SSE}{2}+\beta_2 \right) & \mbox{ if }\sigma_e^2>\sigma_e^*
 \end{array}
 \right. \; .
 \end{eqnarray*}
The points $\sigma_\phi^*$ and $\sigma_e^*$ are the intersection points of the two inverse gamma densities determined by
\[
\sigma^*=\frac{b_1 - b_2}{a (\log b_1-\log b_2)}
\]
where $a$, $b_1$ and $b_2$ are the parameters of the two inverse gamma distributions.

We can define
\begin{align*}
\varepsilon &= \int _{\R^{q+1}\times \R_+^2 } f(\xi | \theta)g(\theta) d\xi d\theta = \int_{\R_+^2}  g(\theta)d\theta \\
& = \int_{\R_+} g_1(\sigma_\phi^2)d\sigma_\phi^2 \int_{\R_+}g_2(\sigma_e^2)d\sigma_e^2
\end{align*}
and density $q(\xi,\theta) = \varepsilon^{-1} g_1(\sigma_\phi^2)
g_2(\sigma_e^2) f(\xi | \theta)$, then \eqref{eq:sty minor} holds.
Since this minorization condition holds for any $d>0$, this establishes
the associated minorization condition.

Simulating the split chain requires evaluation of \eqref{eq:delta
equals 1} similar to the calculation in Appendix~\ref{app:Gibbs},
\begin{align*}
Pr \left( \delta' | \xi', \theta', \xi, \theta \right) & = \frac{ \varepsilon \; q ( \theta, \xi )}{ k \left( \xi , \theta | \xi' , \theta' \right) } \\
% & = \frac{ \varepsilon \; \varepsilon^{-1} \; g_1(\sigma_\phi^2) \; g_2(\sigma_e^2) \; f ( \xi | \theta ) }{ f(\sigma_\phi^2 | \xi') f(\sigma_e^2 | \xi') f(\xi | \theta) } \\
& = \frac{ g_1(\sigma_\phi^2) }{ f(\sigma_\phi^2 | \xi') } \frac { g_2(\sigma_e^2) }{ f(\sigma_e^2 | \xi') } \\
& = \left[ \frac{ I_{\{ \sigma_\phi^2 < \sigma_\phi^{*} \} } (d-K)/\delta_1 + 2 \beta_1 }{ w_1(\xi') + 2 \beta_1 } \right]^{\frac{q}{2} + \alpha_1 } \\
& \cdot \exp \left \{ - \frac{ I_{\{ \sigma_\phi^2 < \sigma_\phi^{*} \} } (d-K)/\delta_1 }{2 \sigma_\phi^2 } + \frac{ w_1(\xi') }{2 \sigma_\phi^2 } \right \} \\
& \cdot \left[ \frac{ I_{\{ \sigma_e^2 < \sigma_e^{*} \} } (d-K)/\delta_2 + SSE + 2 \beta_2 }{ w_2(\xi') + SSE + 2 \beta_2 } \right]^{\frac{M}{2} + \alpha_2 } \\
& \cdot \exp \left \{ - \frac{ I_{\{ \sigma_e^2 < \sigma_e^{*} \} } (d-K)/\delta_2 }{2 \sigma_e^2 } + \frac{ w_2(\xi') }{2 \sigma_e^2 } \right \} \; .
\end{align*}

\end{appendix}

\bibliographystyle{apalike}
%\bibliography{ref}

%%%%%%%%%%%%%%%%%%%%%%%%%%%%%%%%%%%%%%%%%%%%%%%%%%%%%%%%%%%%%%%%%%%%%%%%%%%%%%

%%%%%%%%%%%%%%%%%%%%%%%%%%%%%%%%%%%%%%%%%%%%%%%%%%%%%%%%%%%%%%%%%%%%%%%%%%%%%%
%%%%%%%%%%%%%%%%%%%%%%%%%%%%%%%%%%%%%%%%%%%%%%%%%%%%%%%%%%%%%%%%%%%%%%%%%%%%%%
\end{document}